\documentclass[letterpaper, DIV=15]{scrartcl}  	% use "amsart" instead of "article" for AMSLaTeX format
\usepackage{amssymb, amsthm, mathtools,bbm,color}
\usepackage[square,sort,comma,numbers]{natbib}
\usepackage{enumerate}
\usepackage{dsfont,times,tikz} 
\usetikzlibrary{matrix,arrows.meta,shapes.geometric,calc}
\definecolor{Plum}{rgb}{.5,0,1}

\newcommand{\ran}{\operatorname{ran}}

\newcommand{\spec}{\operatorname{spec}}
\newcommand{\supp}{\operatorname{supp}}

\newcommand{\spa}{\operatorname{span}}

 \def\1{{\mathchoice {\mathrm{1\mskip-4mu l}} {\mathrm{1\mskip-4mu l}} %
		{\mathrm{1\mskip-4.5mu l}} {\mathrm{1\mskip-5mu l}}}}

\newcommand{\bR}{{\mathbb R}}
\newcommand{\bC}{{\mathbb C}}
\newcommand{\bN}{{\mathbb N}}
\newcommand{\bZ}{{\mathbb Z}}

\newcommand{\cH}{{\mathcal H}}
\newcommand{\braket}[2]{\langle {#1} \ | \ {#2}\rangle}
\newcommand{\ket}[1]{|{#1}\rangle}
\newcommand{\ketbra}[1]{\vert #1\rangle\langle #1\vert}
\newcommand{\vp}{\varphi}
\newcommand{\caG}{\mathcal{G}}

\newcommand{\caC}{\mathcal{C}}

\newcommand{\cG}{\mathcal{G}}
\newcommand{\cB}{\mathcal{B}}

\newcommand{\cD}{\mathcal{D}}
\newcommand{\cE}{\mathcal{E}}

\newcommand{\cS}{\mathcal{S}}
\newcommand{\cR}{\mathcal{R}}
\newcommand{\cT}{\mathcal{T}}
\newcommand{\cK}{\mathcal{K}}
\newcommand{\cV}{\mathcal{V}}
\newcommand{\per}{\mathrm{per}}
\newcommand{\dom}{\mathrm{dom}}

\renewcommand{\subset}{\subseteq}

\newcommand{\be}{\begin{equation}}
\newcommand{\ee}{\end{equation}}

\newtheorem{theorem}{Theorem}%[section]
\newtheorem{lem}[theorem]{Lemma}
\newtheorem{cor}[theorem]{Corollary}

\newtheorem{proposition}[theorem]{Proposition}

\numberwithin{equation}{section}
\numberwithin{theorem}{section}
\usepackage{graphicx}				% Use pdf, png, jpg, or eps§ with pdflatex; use eps in DVI mode
\usepackage{amssymb}

\title{\LARGE A Bulk Spectral Gap in the Presence of Edge States for a Truncated Pseudopotential}
\author{Simone Warzel and Amanda Young}
\date{\small\today}							% Activate to display a given date or no date

\begin{document}
\maketitle

\minisec{Abstract} 
We study the low-energy properties of a truncated Haldane pseudopotential with maximal half filling, which describes a strongly correlated system of spinless bosons in a cylinder geometry. 
 For this Hamiltonian with either open or periodic boundary conditions, we prove a spectral gap above the highly degenerate ground-state space which is uniform in the volume and particle number.
 Our proofs rely on  identifying invariant subspaces to which we apply gap-estimate methods previously developed only for quantum spin Hamiltonians.  
 In the case of open boundary conditions, the lower bound on the spectral gap accurately reflects the presence of edge states, which do not persist into the bulk. Customizing the gap technique to the invariant subspace, we avoid the edge states and establish a more precise estimate on the bulk gap in the case of periodic boundary conditions.   
%the martingale method and finite-volume criteria 
%\bigskip
%\bigskip

{\small \tableofcontents}
%\bigskip

\section{Introduction}\label{sec:intro}

Laughlin wavefunctions
\[
\Psi_p(z_1,\dots , z_N) \propto \prod_{1\leq j < k \leq N } (z_j-z_k)^{p+2} \, \prod_{j=1}^N \exp\left( -  \frac{|z_j|^2}{2\ell^2}\right) , 
\]
describe the ground-state properties of highly correlated quantum systems such as quantum Hall systems \cite{PG90,Girvin2005} or rapidly rotating Bose gases \cite{Regnault:2004,Cooper:2008}  in a two-dimensional complex geometry $ z_1,\dots z_N \in \mathbb{C} $. In that context, $ \ell > 0 $ is the magnetic length, which arises naturally  in Hall systems through the perpendicular, constant magnetic field. In the case of dilute Bose gases, the rotational velocity  takes the role of  the magnetic field. 
In his seminal paper \cite{PhysRevLett.51.605}, Haldane derived Hamiltonians,
%\[
$ W_p =  \sum_{1\leq j < k \leq N } w_p(j,k) $,  %\]
with non-negative pair interactions  $ w_p \geq 0 $,  which 
have the Laughlin wavefunction with parameter $ p\in \mathbb{N}_0 $ 
among its zero-energy eigenstates. 
These so-called pseudopotentials also effectively describe the excitations above the Laughlin state. 
The statistics of the many-particle Hilbert space on which $ W_p $ acts is tied to $ p $: bosonic statistics for $ p $ even and fermionic for $ p $ odd. 
The pair potential $ w_p $ projects onto states in the lowest Landau level (LLL) with relative angular momentum at most $ p  $, and formally results from an expansion of a radially symmetric pair interaction  with respect to relative angular momentum; see~\cite{Trugman:1985lv,Girvin2005,LeePapicThomale:2015}. 
A rigorous justification of the emergence of such pair interactions in a scaling limit can be found in  \cite{Lewin:2009,seiringer:2020}. In the case $ p = 0 $, which models a rapidly rotating dilute Bose gas and is the guiding example in this paper, $ w_0 \propto \delta  $  is just a delta-pair interaction on the LLL.

Haldane pseudopotentials are conjectured to faithfully describe all important features and, in particular, the rigidity of quantum Hall systems  or  rotating Bose gases \cite{Girvin2005,Cooper:2008,Lewin:2009,RSY:2014,seiringer:2020}. Their zero-energy eigenstates have a maximal filling fraction $\nu(p) = (p+2)^{-1} $. Higher fillings $ \nu $ lead to a ground-state energy which increases with $ \nu $. For the bosonic case $ p = 0 $ in the planar geometry, this results in the Yrast line of ground-state energies as a function of the conserved total angular momentum; see \cite{Regnault:2004,Cooper:2008,Lewin:2009}.  
Most importantly, pseudopotentials are conjectured to have a uniform spectral gap above its ground-state space -- a feature, which is responsible for the incompressibilty of the quantum fluid \cite{Lieb:2019vl,Roug19,NWY:2021} as well as the quantization of the Hall conductance~\cite{hastings:2015,Bachmann:2018lb,Bachmann:2021dp}. The gap is expected to be stable with respect to perturbations and the details of the two-dimensional complex geometry (cf.\ \cite{Haldaneetal2016}). 

In this paper, we follow the route taken in~\cite{Rezayi:1994wn,Bergholtz:2005pl,Jansen:2012da,Nakamura:2012bu,NWY:2020,NWY:2021} and simplify matters by changing the geometry and truncating the pseudopotential. The 
cylinder geometry has the advantage that its LLL is spanned by an orthonormal basis $\{\psi_x | x\in\bZ\}$ with a natural one-dimensional lattice structure.
The spanning one-particle orbitals  are given by
\be\label{eq:Landauorbital}
\psi_{x}(\xi,\eta) = \sqrt{\frac{\alpha}{2\pi^{3/2}\ell^2}}
\exp\left(i  x \frac{\alpha \eta}{\ell}\right) \exp\left(-\frac{1}{2}\left[\frac{\xi}{\ell}-x\alpha\right]^2\right)
\ee
with $\xi \in \bR$, $\eta\in[0,2\pi R)$ marking the positions on the cylinder, and $\alpha:=\ell/R$ the ratio of the magnetic length to the cylinder radius $R> 0 $. In terms of the complex coordinates $ z_j := \xi_j + i \alpha \eta_j  $ the Laughlin wavefunctions in this geometry take the form
\[
\Psi_p(z_1,\dots , z_N) \propto \prod_{1\leq j < k \leq N } \left(e^{z_j/R} -e^{z_k/R} \right)^{p+2} \, \prod_{j=1}^N \exp\left( -  \frac{|\xi_j|^2}{2\ell^2}\right) . 
\]
The pair interaction of a pseudopotential, which has this Laughlin state in its zero-energy eigenspace, projects onto orbitals with relative coordinates $ |x_j - x_k | \leq p $. 
Using the annihilation and creation operators $ a_x $ and $ a_x^* $ of the one-particle orbitals \eqref{eq:Landauorbital}, whose statistics is again determined by $ p \in \mathbb{N}_0 $,   the pseudopotential is of the form
\[
W_{p} = \sum_{s\in\bZ/2}B_{p,s}^*B_{p,s}, \quad \text{with}\quad B_{p,s} := \sum_{k}{\vphantom{\sum}}' F_p(2k\alpha)a_{s-k}a_{s+k}\quad \mbox{and}\quad F_p(t) := \sum_{0\leq m \leq p } H_m(t) e^{-t^2/4} .
\]
The primed sum is over $\bZ$ if $s$ is integer and $\bZ+\frac{1}{2}$ otherwise, and the summation in $ F_p $ is over integers $ m $ of the same parity as $ p $. Depending on the parity of $ p $, the real polynomials $ H_m $ result from orthogonalizing the even, respectively odd, monomials in $ \{ 1, t , \dots , t^p \} $ with respect to the natural scalar product induced by the $ k $-sum; for details see~\cite{LeePapicThomale:2015}. In the thin-cylinder limit $ \alpha \to \infty $, $ H_m $ is the $m$-th order Hermite polynomial. We refer to~\cite{Rezayi:1994wn,Lee:2004bs,Jansen:2012da} and in particular \cite{LeePapicThomale:2015} for a more detailed discussion of pseudopotentials in the cyclinder geometry. 

As a case study of a bosonic problem, we focus on the simplest case $ p = 0 $ in the thin-cylinder limit for which we may take $ H_0(t) = 1 $.  Analogous to the $ p = 1 $ fermionic case studied in~\cite{NWY:2020,NWY:2021,WY:2021b}, for  $ \alpha \to \infty $ it is reasonable 
to truncate the summation in $ B_{p,s} $ to its lowest-non-trivial order, that is, $|k|\leq 1 $ for $ p= 0 $. This results in a finite-range model which, after changing the prefactor, coincides with the formal Hamiltonian 
\begin{equation}\label{eq:Haminf}
 \sum_{x} n_x n_{x+1} + \kappa  \sum_{x} q_x^* q_x , \quad \text{with}\quad  q_x = a_x^2 -\lambda a_{x-1} a_{x+1}, \quad n_x = a_x^* a_x ,
\end{equation}
and $ \kappa = e^{\alpha^2/2}/ 4 $ and $ \lambda = - 2 e^{-\alpha^2} $ as the physical parameters. For the purposes of this work, we consider more generally that $\kappa>0$ and $\lambda\in\bC\setminus\{0\}$.  The aim of this paper is to address the rigidity properties in this truncated bosonic model. In particular, we will establish a uniform spectral gap and a bound on the analogue of the Yrast line in this simplified model. \\

\subsection{Main results}\label{sec:main}
For the precise mathematical definition of the Hamiltonian analyzed in our results, we restrict the truncated model~\eqref{eq:Haminf} to Landau orbitals~\eqref{eq:Landauorbital} whose center variable $ x \in \mathbb{Z} $ is in an interval $ \Lambda  = [ a,b] $. 
The bosonic Fock space associated with these Landau orbitals is the closure
\begin{equation}
\cH_\Lambda := \overline{\spa}\left\{ | \mu \rangle \, | \, \mu \in \mathbb{N}_0^\Lambda \right\} 
\end{equation}
of orthonormal vectors associated with occupation numbers $ \mu \in \mathbb{N}_0^\Lambda $ of the single-particle orbitals $ x \in \Lambda $.  
We will refer to $ \mu $ as a particle configuration. 
The Fock space $ \mathcal{H}_\Lambda $ carries the natural scalar product $ \langle \varphi | \psi \rangle := \sum_{\mu \in \mathbb{N}_0^\Lambda } \overline{\varphi(\mu)} \psi(\mu) $ where $\psi = \sum_{\mu\in\bN_0^\Lambda}\psi(\mu)\ket{\mu}$. 
The truncated Hamiltonian with open respectively periodic boundary conditions, then corresponds to the energy form
\[
	\langle \psi | H^{\sharp}_\Lambda \psi \rangle  = \sum_{\mu \in \mathbb{N}_0^\Lambda  }   e_\Lambda^\sharp(\mu)   |\psi(\mu)|^2 + \kappa \, \sum_{\nu \in \mathbb{N}_0^\Lambda} \sum_{x\in \Lambda^\sharp}  \left| (q_x\psi)(\nu) \right|^2 , \quad \sharp \in \{ \textrm{obc} , \textrm{per} \} ,
%	 \langle \psi | H_\Lambda^\textrm{per} \psi \rangle & = \sum_{\mu \in \mathbb{N}_0^\Lambda  }   e_\Lambda^\textrm{per}(\mu)   |\psi(\mu)|^2 + \kappa \, \sum_{\nu \in \mathbb{N}_0^\Lambda} \sum_{x=a}^{b} \left| (q_x\psi)(\nu) \right|^2 
\]
with $ e_\Lambda^\textrm{obc}(\mu) := \sum_{x=a}^{b-1} \mu_x \mu_{x+1} $, respectively, $ e_\Lambda^\textrm{per}(\mu) := e_\Lambda^\textrm{obc}(\mu)  + \mu_b \mu_{a} $,  representing the electrostatic energy. The summation for the hopping term extends over $ \Lambda^\textrm{obc} = [a+1,b-1] $ for open boundary conditions and $ \Lambda^\textrm{per} = [a,b] $  for periodic boundary conditions, in which case additions are understood modulo the volume $ |\Lambda| = b-a+1 $. 
The hopping operator is defined via
\begin{equation}\label{creation-action} 
	(q_x\psi)(\nu) := \sqrt{(\nu_x+1)(\nu_x+2)} \  \psi\left((\alpha^*_x)^2 \nu\right) - \lambda   \sqrt{(\nu_{x-1}+1)(\nu_{x+1}+1)} \ \psi\left(\alpha^*_{x+1} \alpha^*_{x-1} \nu\right) 
\end{equation}
and expressed in terms of the functions $ \alpha^*_x : \mathbb{N}_0^\Lambda  \to \mathbb{N}_0^\Lambda  $ with $ x \in \Lambda $, which map configurations $ \nu $ to  $ \alpha^*_x \nu  $  by adding a particle at the site $ x \in \Lambda $, i.e.\ $ \nu_x \to \nu_x+1 $ and the particle numbers at all other sites are unchanged. We will also use $ \alpha_x : \mathbb{N}_0^\Lambda \to   \mathbb{N}_0^\Lambda  $ for subtracting a particle from the site $ x \in \Lambda $, that is $ \nu_x \to \nu_x - 1 $, provided that $ \nu_x \geq 1 $.

Through Friedrich's extension theorem, the above non-negative energy forms define  (unbounded) self-adjoint operators $ H_\Lambda^\sharp : \dom(H_\Lambda^\sharp) \to \cH_\Lambda  $.  To ease the notation, we will also frequently drop the superscript $ '\textrm{obc}' $ for open boundary conditions and write $ H_\Lambda \equiv H_\Lambda^\textrm{obc} $. Both Hamiltonians are frustration free as they are sums of non-negative terms with ground-state spaces given by the respective kernels
\[ \caG_\Lambda^\sharp := \ker H_\Lambda^\sharp . \]  
For open boundary conditions
$ \caG_\Lambda \equiv \caG_\Lambda^\textrm{obc} $
can easily be seen to be infinite dimensional as, e.g. $H_\Lambda \ket{n0\ldots0}=0$ for all $n\in\bN_0$. One of the key results, which is contained in Section~\ref{sec:VMD}, is an explicit characterization of  $  \caG_\Lambda $ and of $  \caG_\Lambda^\textrm{per} $, whose dimension will be shown to grow exponentially with the volume.  Building on this, the main aim in this work is to prove that the spectral gaps above the ground states are strictly positive uniformly in the system size $ |\Lambda | $, i.e. for both boundary conditions $ \sharp \in \{ \textrm{obc} , \textrm{per} \} $  there exists $\gamma^\sharp>0$ so that
\begin{equation}
E_1^\sharp(\cH_\Lambda) := \inf_{0\neq\psi\in \dom(H_\Lambda^\sharp) \cap (\cG_{[1,L]}^\sharp)^\perp } \frac{\braket{\psi}{H_\Lambda^\sharp\psi}}{\|\psi\|^2} \geq \gamma^\sharp
\end{equation}
for all intervals $\Lambda$ sufficiently large.  
We will 
estimate this spectral gap in terms of
\begin{equation}
\gamma_\kappa^\textrm{obc}(|\lambda|^2)  := \frac{1}{5}\min\left\{  4\gamma_\kappa^\textrm{per}(|\lambda|^2)  ,  \frac{2\kappa |\lambda|^2}{\kappa+1}  \right\}, \quad  \gamma_\kappa^\textrm{per}(|\lambda|^2)  :=  \frac{1}{4} \min\left\{1,  \frac{2\kappa}{\kappa+1} ,    \frac{2 \kappa}{1+\kappa|\lambda|^2}    \right\} .
\end{equation} 
%and the spectral gap above the ground state of $H_{[1,6]}^\textrm{obc} $ restricted to the invariant subspace 
%\be
%\label{small_gap}
%\spa\{\ket{101010},\,\ket{100200},\;\ket{020010}\},\ee
%which is equal to $2\kappa.$  
For open boundary conditions the main result is the following:
%
%and the spectral gap of the restriction of $ H_{[1,L]} $ (with $L=5,6$) to an invariant, finite-dimensional subspace $ \caC_{[1,L]}^\infty \subset \dom(H_{[1,L]} ) $:
%\be\label{eq:gapinvariant} 
%E_1(\caC_{[1,L]}^\infty) := \inf_{0\neq\psi\in \caC_{[1,L]}^\infty\cap \cG_{[1,L]}^\perp } \frac{\braket{\psi}{H_{[1,L]}\psi}}{\|\psi\|^2}.
%\ee
%The fact that $ \dim  \caC_{[1,L]}^\infty < \infty $ will ensure $  E_1(\caC_{[1,L]}^\infty) > 0 $. Postponing the precise definition of this finite-dimensional subspace to~\eqref{def:IVBVMD} below, our main result reads as follows:
%
\begin{theorem}[OBC Spectral Gap]\label{thm:main}
 There is a monotone increasing function $ f : [0,\infty) \to [0,\infty) $ such that for all $ 0\neq \lambda \in \mathbb{C} $  with the property $ f(|\lambda|^2)   < 1/3 $ and all $ \kappa \geq 0 $:
\be\label{eq:main}
\inf_{\Lambda \subset \mathbbm{Z}, |\Lambda| \geq 10 } E_1^\textrm{obc}(\cH_\Lambda)   \geq \min\left\{ \gamma_\kappa^\textrm{obc}(|\lambda|^2) , \, \frac{2\kappa}{3} \left( 1 - \sqrt{3 f(|\lambda|^2/2) } \right)^2\right\} .  
\ee
\end{theorem}
The proof of this theorem is provided in~Subsection~\ref{sec:eeproofmain}. An explicit expression for $ f $, which is monotone increasing, is stated in Theorem~\ref{thm:tiling_gap}, where we also show that $f(|\lambda|^2/2)<1/3$ for $|\lambda|\leq 7.4$.

For $\kappa>0$ fixed and $|\lambda| \ll 1$, which covers the physical parameter regime, the minimum in~\eqref{eq:main} is taken at $ \gamma_\kappa^{\textrm{obc}}(|\lambda|^2)  $ which 
 is of the order $\mathcal{O}(|\lambda|^2)$. The bound is sharp in this regime due to the existence of edge states which are discussed in more detail at the end of Subsection~\ref{sec:Electrostatic}. %which resemble ground-states on the interior of the interval. 
 %The ground states are described in detail in Section~\ref{sec:VMD}, and discuss the invariant subspaces that contain these edge states in more detail following the proof of Theorem~\ref{thm:perp_bound}. 
 As an example of such an edge state, consider the two-dimensional space 
\[\spa\{\ket{20100\ldots0}, \, \ket{1200\ldots0}\}\subseteq \cH_\Lambda\]
which is invariant under the action of $H_\Lambda^\textrm{obc}$. Diagonalizing the associated $ 2\times 2 $ matrix yields eigenvalues $E_{\pm} = (\kappa|\lambda|^2+\kappa + 1)(1 \pm \sqrt{1-4\kappa|\lambda|^2/(\kappa|\lambda|^2+\kappa + 1)^2})$, the smallest of which is of order
\be\label{ex:edge_energy}
E_- = \frac{2\kappa|\lambda|^2}{\kappa+1} + \mathcal{O}(|\lambda|^4)
\ee
when $|\lambda|\ll 1$. 
In contrast, the bulk gap is strictly bounded away from zero uniformly for small $ |\lambda | $. This is shown in our second main result.
% which estimates the spectral gap of the periodic operator $ H_\Lambda^\textrm{per} $  in terms of 
%\begin{equation}
%\gamma_\kappa^\textrm{per}(|\lambda|^2)  :=  \frac{1}{4} \min\left\{1,  \frac{2\kappa}{\kappa+1} ,    \frac{2 \kappa}{1+\kappa|\lambda|^2}    \right\} 
%\end{equation} 
%and the monotone function $ f $, which already appeared in Theorem~\ref{thm:main}.
%
\begin{theorem}[Bulk Spectral Gap]\label{thm:main2}
 There is a monotone increasing function $ f : [0,\infty) \to [0,\infty) $ such that for all $ 0\neq \lambda \in \mathbb{C} $  with the property $ f(|\lambda|^2/2)   < 1/3 $ and all $ \kappa \geq 0 $:
\be\label{eq:main2}
\liminf_{|\Lambda| \to \infty} E_1^\textrm{per}(\cH_\Lambda)   \geq  \min\left\{ \gamma_\kappa^\textrm{per}(|\lambda|^2) , \,  \frac{\kappa}{3 (1+|\lambda|^2) } \left( 1 - \sqrt{3 f(|\lambda|^2/2) } \right)^2 \right\} .  
\ee
\end{theorem}
The proof of this theorem is given in Subsection~\ref{Sec:ProofMain2}. As explained next in detail, the key to establishing this result is to 
 explicitly deconstruct the Hilbert space into invariant subspaces to which different gap estimating techniques are applied to circumvent the edge states. In contrast to~\eqref{eq:main} the bound~\eqref{eq:main2}  survives the limit $ \lambda \to 0 $, in which case  the (bulk) spectral gap is explicit  $ \min\{ 1, 2 \kappa \} $.  
 For a more detailed understanding of the bulk excitations, we explore  in Subsection~\ref{sec:Excited_States} other invariant subspaces which we conjecture to support the lowest excitations. Their energies, which are of course consistent with~\eqref{eq:main2}, 
   are determined perturbatively for small $ |\lambda| $ in
Subsection~\ref{sec:Excited_States}. We also include there a brief discussion of the many-body scars in this model.

\subsection{Invariant subspaces and the proof strategy}\label{sec:invariant} 
As its fermionic cousin studied in~\cite{NWY:2021}, the bosonic model at hand is not integrable in the sense that there is no extensive number of independent conserved quantities. For the Haldane pseudopotentials, the conserved quantities are the total particle number $ \sum_{x} n_x $ and the center of mass $ \sum_x x n_x $. Regardless, one can explicitly determine an extensive number of invariant subspaces. This observation is the foundation of the analysis in this paper. A preview of this was provided above where we %computed of the gap in \eqref{small_gap} and 
discussed the edge-state example.

We recall from \cite{Birman:1987} that a closed subspace $ \cV \subset \cH $ is invariant, or equivalently reducing in the case of a self-adjoint operator $ A: \dom(A) \to \cH $, if and only if the corresponding orthogonal projection $ P_{\cV} $ commutes with the operator,
\[ P_{\cV} A \  \dom(A) = A P_{\cV}  \ \dom(A) . \]
Since the electrostatic part of the Hamiltonian is diagonal in the configuration basis, we can construct an invariant subspace of $H_\Lambda^\sharp$ by considering the action of the hopping terms on a fixed configuration $\sigma_\Lambda(R)\in\bN_0^\Lambda$. Namely, an invariant subspace results from taking the span of all configuration states  $\mu\in\bN_0^\Lambda$ that have a nonzero inner product with a state of the form $(q_{x_k}^*q_{x_k}\ldots q_{x_1}^*q_{x_1})\ket{\sigma_\Lambda(R)}$ for some $k\geq 1$ and $x_1,\ldots,x_k $.
Similar to the analysis from~\cite{NWY:2021}, a convenient way of labeling the spanning set of configurations is by means of domino tilings of $ \Lambda $, for which the generating configuration $ \sigma_\Lambda(R) $ is characterized by a root tiling $R$. %Since our constructions  involves boundary  as well as void, monomer and dimer dominos, we refer to the tilings used in this paper as \emph{BVMD-tilings}. 
The exact definition of these lattice tilings is in Section~\ref{sec:VMD}, where we also define and state the key properties of the associated closed, invariant subspace $ \mathcal{C}_\Lambda \equiv \mathcal{C}_\Lambda^\textrm{obc}$  and $  \mathcal{C}_\Lambda^\textrm{per} $ of all tiling states for open and periodic boundary conditions, respectively; see~\eqref{tiling_spaces} and~\eqref{per_tiling_space}.  Most importantly, $\caC_\Lambda^\sharp $ for both $ \sharp  \in \{ \textrm{obc}, \textrm{per} \}  $ contains the ground-state space $ \mathcal{G}_\Lambda^\sharp $ of the respective Hamiltonian $H_\Lambda^\sharp$.
Since both
$$ \cH_\Lambda = \caC_\Lambda\oplus \caC_\Lambda^\perp , \quad\mbox{and}\quad  \cH_\Lambda = \caC_\Lambda^\textrm{per} \oplus  \left( \caC_\Lambda^\textrm{per}\right) ^\perp $$ 
constitute  orthogonal decompositions of the Hilbert space into closed invariant subspaces of $ H_\Lambda^\textrm{obc} $ and $ H_\Lambda^\textrm{per} $ respectively and  $ \caG_\Lambda^\sharp \subset  \caC_\Lambda^\sharp \subset \dom H_\Lambda^\sharp $, the spectral gap of $H_\Lambda^\sharp $ can be realized as
\begin{equation}\label{eq:twoway}
E_1^\sharp(\cH_\Lambda) = \min\left\{ E_1^\sharp(\caC_\Lambda^\sharp) , E_0^\sharp\left( \big(\caC_\Lambda^\sharp\big)^\perp\right)\right\} , \quad \sharp \in \{ \textrm{obc}, \textrm{per} \} , 
\end{equation}
where $E_1^\sharp(\caC_\Lambda^\sharp)$ is the spectral gap of $ H_\Lambda^\sharp $  restricted to  $ \caC_\Lambda^\sharp$,  and $E_0^\sharp\left( \big( \caC_\Lambda^\sharp\big) ^\perp \right)$ is the ground-state energy of $ H_\Lambda^\sharp $ in the orthogonal subspace $ \big( \caC_\Lambda^\sharp\big) ^\perp  $, i.e.
\begin{equation}\label{eq:gapabbrev}
E_1^\sharp(\caC_\Lambda^\sharp) := \inf_{0\neq\psi\in \caC_\Lambda^\sharp\cap \left(\caG_\Lambda^\sharp\right)^\perp} \frac{\braket{\psi}{H_\Lambda^\sharp\psi}}{\|\psi\|^2}, \qquad 
E_0^\sharp\left(\big(\caC_\Lambda^\sharp\big)^\perp\right) := \inf_{0\neq\eta\in \left(\caC_\Lambda^\sharp\right)^\perp\cap\dom(H_\Lambda^\sharp)} \frac{\braket{\eta}{H_\Lambda^\sharp\eta}}{\|\eta\|^2} . 
\end{equation}
We employ different strategies to lower bound these energies uniformly in $ \Lambda $:
\begin{enumerate}\itemsep0ex
\item The martingale method for a bound on $ E_1^\textrm{obc}(\caC_\Lambda)  $ (cf.~Section~\ref{sec:MM}).
\item A finite-volume condition, which lower bounds the periodic gap $ E_1^\textrm{per}(\caC_\Lambda^\textrm{per}) $ in terms of the spectral gap $ E_1^\textrm{obc}(\caC_\Lambda^\infty)  $ for open boundary conditions restricted to  the subspace of bulk tilings $ \caC_\Lambda^\infty \subset \caC_\Lambda $ (cf.~\eqref{def:IVBVMD}  and~Section~\ref{sec:UBG}). 
\item Electrostatic estimates for bounds on $ E_0^\textrm{obc}(\caC_\Lambda^\perp) $ and $  E_0^\textrm{per}\left(\big(\caC_\Lambda^\textrm{per}\big)^\perp\right) $, which also relate to the Yrast line mentioned in the introduction (cf.~Theorems~\ref{thm:perp_bound} and~\ref{thm:electro2}, and Proposition~\ref{prop:Yrast}). 
\end{enumerate}

Before delving into the details, let us put these strategies in context. 

The martingale method and finite-volume criteria \cite{affleck:1988,knabe:1988,nachtergaele:1996,Gosset:2016hy,lemm:2018,Lemm:2019qh,Kastoryano:2019ja,abdul-rahman:2020} have previously only been developed for and applied to quantum spin or lattice fermion systems, for which the dimension of the finite-volume Hilbert space is finite.  For our lattice bosons, the dimension of $  \cH_\Lambda $ and even of $  \caC_\Lambda $ is infinite. This does not merely require technical amendments of the method, but poses the additional problem that the method's  induction hypothesis, namely the existence of a positive spectral gap for any finite-volume Hamiltonian, does not a priori hold. For the present model, we solve this by showing that $ E_1^\textrm{obc}(\caC_\Lambda)   $ is realized on the finite-dimensional invariant subspace of bulk tilings $ \caC_\Lambda^\infty \subset  \caC_\Lambda $; see~\eqref{def:IVBVMD} and Theorem~\ref{thm:gap_reduction}.  

Similar to, but in fact more severe than its fermionic cousin studied in~\cite{NWY:2021}, the present model has plenty of low-energy edge states for the Hamiltonian $ H_\Lambda $ with open boundary conditions in comparison to the bulk Hamiltonian $ H_\Lambda^\textrm{per} $. In this situation, it is a well recognized hard problem to rigorously establish a bulk gap which does not scale with the energy of edge states. 
This stems from the fact that the known proof strategies, the martingale method and finite-volume criteria, involve finite-volume Hamiltonians with open boundary conditions. 
We solve this problem by restricting these proof techniques a priori to invariant subspaces $\caC_\Lambda^\textrm{per} \subset  \caC_\Lambda^\infty $, which project out the edge states. This novel twist on these methods is provided here for the truncated bosonic Haldane pseudopotential. However, it is equally applicable to the fermionic  $\nu =1/3$ model. By appropriately modifying the approach here, one can prove a bulk gap which is stable for small $ |\lambda| $ for the analogously truncated model thereby improving~\cite[Therorem~1.2]{NWY:2021}. This analysis is carried out in the subsequent work \cite{WY:2021b}.

Hence, despite many similarities to the fermionic case of the truncated $\nu = 1/3$ Haldane pseudopotential studied in \cite{NWY:2021}, beyond modifying and streamlining of the proof of that result, the analysis in the present paper tackles three additional challenges -- the adaptation of the martingale method through a reduction to finite-dimensional subspaces, electrostatic estimates, and a proof of a bulk gap that circumvents edge states by customizing the gap-techniques to appropriate invariant subspaces. %A sufficient condition for 
%\[ \spa\left\{ | \mu \rangle \, | \, \mu \in  \cB_\Lambda \right\} \]
%to be a reducing subspace of $ H_\Lambda $ is that $ q^*_xq_x |\nu \rangle \rangle \in \spa\left\{ | \mu \rangle \, | \, \mu \in  \cB_\Lambda \right\}  $ for all  $ \nu \in  \cB_\Lambda $ and all $ x \in \Lambda $. 

\section{Tilings and their state spaces}\label{sec:VMD}

%We introduce a number of invariant subspaces of $H_\Lambda$ that are constructed using lattice tilings, called Boundary-Void-Monomer-Dimer (BVMD) subspaces. These subspaces are pairwise mutually orthogonal. Moreover, each BVMD space supports a single nonzero ground state of the Hamiltonian, and the collection of these vectors constitutes an orthogonal basis for the ground state space of $H_\Lambda$ for intervals of length $|\Lambda|\geq 5$.

%\textbf{The goal of this section is to identify an invariant subspace that contains the ground state space $ \mathcal{G}_\Lambda = \ker H_\Lambda$. In fact, the invariant subspace is constructed as a direct sum of invariant subspaces which each support a unique ground state. The definition of these subspaces employs the occupation basis $ \{ \ket{\mu} \} $  defining $ \cH_\Lambda $ and uses domino-tilings of the lattice, which indicate the occupation number at sites covered by the domino. Since the Hamiltonian is frustration-free, any particle configuration $ \mu $ gives rise to a ground state of the electrostatic term $ e_\Lambda(\mu) $ if and only if $\mu_x\mu_{x+1}=0$ for all $x\in \Lambda $.

The goal of this section is to identify  invariant subspaces $\caC_\Lambda$ and $ \caC_\Lambda^\textrm{per} $  that contain the ground-state space $ \mathcal{G}_\Lambda = \ker H_\Lambda$ and $  \mathcal{G}_\Lambda^\textrm{per} = \ker H_\Lambda^\textrm{per}$ respectively. These subspaces will be constructed as a direct sum of invariant subspaces $\caC_\Lambda^\sharp(R)$ each of which supports a unique ground state and is spanned by a finite subset of the orthonormal occupation basis $ \{ \ket{\mu} : \mu\in\bN_0^\Lambda\} $ of $ \cH_\Lambda $. Each of the chosen occupation states is described by a domino-tiling of the lattice, where the values of each domino indicate the occupation numbers of the covered sites.
 
 To motivate the definition of these tiles, recall that as the Hamiltonian is frustration-free, the ground state space is the set of vectors that simultaneously minimize the energy of all interaction terms: 
 \[
 \ker(H_\Lambda) = \bigcap_{x=a}^{b-1}\ker(n_xn_{x+1}) \cap \bigcap_{x=a+1}^{b-1}\ker(q_x), \qquad  \ker(H_\Lambda^\textrm{per}) = \bigcap_{x=a}^{b}\left( \ker(n_xn_{x+1}) \cap \ker(q_x) \right) .
 \]
 Every particle configuration $ \ket{\mu }$ gives rise to an electrostatic energy and is a ground state of these terms if and only if $\mu_x\mu_{x+1}=0$ for all $x$. The operator $q_x$ acts nontrivially on the sites $ \{ x-1,x,x+1\} $, and satisfies the equation
\begin{equation}\label{dipole_gs}
	q_x\ket{101} = -\frac{\lambda}{\sqrt{2}} q_x\ket{020}.
\end{equation}
Therefore, starting from a configuration of 1's and 0's that is a ground state of the electrostatic terms, a ground state of the hopping terms $\sum_{x} q_x^*q_x$ in either case of boundary conditions can be constructed by summing over the set of all configurations obtained from replacing sequences (101) with (020) and appropriately scaling. A relation similar to \eqref{dipole_gs} holds if either the first or third site in the configuration on the LHS contains more than one particle or  if the middle site on the RHS of \eqref{dipole_gs} contains more than 2 particles. 
However, the action of $ q_x $ will result in a configuration with electrostatic energy and thus such configurations cannot contribute to a ground state. This indicates that the bulk of a ground state can be at most half-filled. 
Of course, there are other configurations that satisfy the electrostatic ground-state condition $  e_\Lambda^\sharp(\mu) = 0 $. For example, any configuration with at most one particle is automatically in the kernel of $q_x$. 

With these observations in mind, we now turn to defining Void-Monomer-Dimer tilings. Since \eqref{dipole_gs} must be satisfied at every site $x$ in either the thermodynamic limit $\Lambda \uparrow \bZ$ or on $ \Lambda $ in the periodic geometry, we first define three \emph{bulk tiles}:
\begin{enumerate}
	\item a void $V=(0)$, which covers a single site and contains no particles,
	\item a monomer $M=(10)$, which covers two sites and contains a single particle on the first site, and
	\item a dimer $D=(0200)$, which covers 4 sites and contains only two particles on the second site.
\end{enumerate}
A \emph{VMD tiling} of $\bZ$ is any tiling of the entire lattice by these three tiles. Similarly, a \emph{periodic VMD tiling} of a finite volume $\Lambda=[a,b]$ with the periodic boundary conditions is any covering of the ring by these tiles.

\subsection{BVMD tilings for open boundary conditions}\label{subsec:tilings}

To describe the ground state for open boundary conditions, we need additional \emph{boundary tiles} to account for possible edge configurations that can support the ground states. One way to obtain such tiles is to consider the set of truncated tiles created from restricting a VMD-tiling of $ \bZ $ to $\Lambda$. We ignore cuttings that produce tiles with no particles, as these can be equivalently constructed using voids. This produces the following set of boundary tiles, which we refer to as \emph{$\bZ$-induced boundary tiles}:
\begin{enumerate}
	\item \emph{On the left boundary:} a truncated dimer $B_2^l=(200)$ which covers three sites and contains two particles on the first site.
	\item \emph{On the right boundary:} 
	\begin{enumerate}
		\item a truncated monomer $M^{(1)}= (1)$, which covers one site and contains one particle,
		\item a truncated dimer $B_2^r=(02)$, which covers two sites and contains two particles on the second site, and
		\item a truncated dimer $D^{(1)}=(020)$, which covers three sites and contains two particles on the second site. 
	\end{enumerate}
\end{enumerate}
To acccount for the full ground state of $ H_\Lambda $ two additional types of boundary tiles are found from the following observation: if $\mu\in \bN_0^\Lambda$ is of the form
\[
\ket{\mu} = \ket{\mu_a00}\otimes\ket{\mu'}\otimes\ket{0\mu_b}
\]
where $\mu'\in\{0,1,2\}^{|\Lambda|-5}$ is a particle configuration obtained from a tiling of $[a+3,b-2]$ by bulk tiles (see \eqref{config}) then
\be\label{artificial_action}
H_{[a,b]}\ket{\mu} = \ket{\mu_a00}\otimes H_{[a+3,b-2]}\ket{\mu'}\otimes\ket{0\mu_b}.
\ee
A similar statement holds if $\ket{\mu} = \ket{\mu_a00}\otimes\ket{\mu^b}$ or $\ket{\mu} = \ket{\mu^a}\otimes\ket{0\mu_b}$ where $\mu^b$, resp.\ $\mu^a$, is the particle configuration associated to a VMD-tiling by bulk and right, resp.\ left, $\bZ$-induced boundary tiles. Thus, we introduce the following \emph{non-$\bZ$-induced boundary tiles} for $n\geq 3$:
\begin{enumerate}
	\item \emph{On the left boundary:} $B_n^l = (n00)$ covering three sites with $n$ particles on the first site.
	\item \emph{On the right boundary:} $B_n^r = (0n)$ covering two sites with $n$ particles on the second site.
\end{enumerate}
The fact that an edge site of a ground state of $ H_\Lambda $ can hold an arbitrary number of particles is a consequence of the lack of hopping at the boundary. As indicated by \eqref{dipole_gs}, the interior sites of a ground state can hold at most two particles, and so this completes the set of boundary tiles.

%The set of BVMD-tilings of an interval  $ \Lambda $ will be abbreviated by $ \cT_\Lambda $. It will often be useful to write such tilings in their ordered form
A \emph{BVMD-tiling} of $\Lambda$ is then defined as any ordered covering of $\Lambda$
\be
T = (T_1, \, T_2, \ldots, \, T_k) \in\cT_\Lambda
\ee
where each $T_i$ is one of the tiles defined above and only $T_1$, resp. $T_k$, can belong to the set of left, resp.\ right, boundary tiles. The number of tiles $k$ in a tiling of $\Lambda$ can vary since tiles have different lengths. The set of all BVMD-tilings of $ \Lambda $ will be abbreviated by $ \cT_\Lambda $.

Motivated by~\eqref{dipole_gs} we define two \emph{substitution rules} that allow us to create a new tiling $T'$ from a fixed tiling $T$ by replacing two neighboring monomers by a dimer or vice-versa. Pictorially, these are represented by
\begin{equation}\label{eq:replacement}
	(0200) \leftrightarrow (10)(10) \quad (020) \leftrightarrow (10)(1),
\end{equation}
in which we exchange a bulk dimer $D$ with two bulk monomers, or a truncated right dimer $D^{(1)}$ with a bulk monomer and truncated monomer. These rules induce a equivalence relation ``$\leftrightarrow$'' on the set of BVMD tilings $\cT_\Lambda$. Namely, we say that two tilings $T,T'\in\cT_\Lambda$ are \emph{connected} and write $T\leftrightarrow T'$ if $T$ becomes $T'$ after a finite number of replacements of the form in \eqref{eq:replacement}. Each equivalence class $\cT_\Lambda(R)$ is uniquely characterized by a \emph{root tiling} $R = (R_1, \ldots, R_k)\in\cT_\Lambda$, which is defined as any tiling such that 
\begin{align*}
	R_1 & \in \{V, \, M\}\cup\{B_n^l \, : \, n \geq 2\} , \\
R_i &\in \{V,\, M\} \quad \mbox{for all}  \,\, 1<i<k ,  \\
R_k &\in \{V,\, M,\, M^{(1)} \}\cup\{B_n^r \, : \, n\geq 2 \},
\end{align*}
%Equivalently, a root-tiling can be described by a collection of sites $\cV \subseteq \Lambda$ that hold void tiles, and a pair of boundary conditions $B=(B^l,B^r)$ where
%\[
%B^l \in \{\emptyset, \, B_n^l \, |\,  n\geq 2\}, \quad B^r\in \{ \emptyset, \, M_r, \,B_n^r \, |\,  n\geq 2\}.
%\]
see Figure~\ref{fig:equivalence_class}. Said differently, a root tiling is any BVMD-tiling of $\Lambda$ that does not use the dimers $(0200)$ or $(020)$. We denote the set of root-tilings by $\cR_\Lambda$. 
Consequently, we can partition $\cT_\Lambda$ into subsets labeled by the root tilings,
\begin{equation}
	\cT_\Lambda = \biguplus_{R\in\cR_\Lambda} \cT_\Lambda(R), \quad \cT_{\Lambda}(R) = \{T \in \cT_\Lambda \, | \, T\leftrightarrow R\}.
\end{equation}

\begin{figure}
	\centering{\includegraphics[scale=.3]{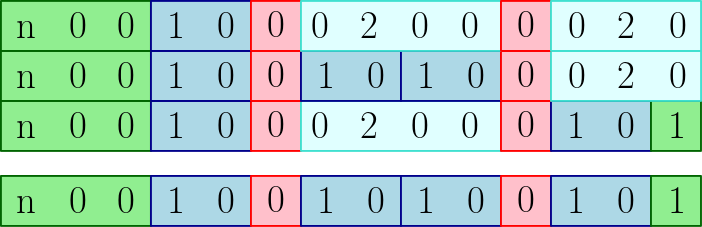}}
	\caption{The equivalence class $\cT_\Lambda(R)$ for the root tiling $R$ in the bottom line with boundary conditions $B_n^l $ with $n\geq2$ and $M^{(1)}$.}\label{fig:equivalence_class}
\end{figure}

The natural embedding 
\begin{equation}\label{config}
	\sigma_{\Lambda}: \cT_\Lambda \to \bN_0^\Lambda, \quad T\mapsto \sigma_\Lambda(T)
\end{equation}
identifies each tiling $T$ with its particle configuration $\sigma_\Lambda(T)$. As we will show next,  the particle configurations in the range are uniquely characterized by the tiling.
\begin{lem}[BVMD Tiling Configurations]\label{lem:tiling_configs}
	Fix an interval $\Lambda = [a,b]$ with $|\Lambda|\geq 4$. A configuration $\mu\in\bN_0^{\Lambda}$ is in $\ran \sigma_\Lambda$ if and only if the following three conditions hold:
	\begin{enumerate}
		\item $\mu_x \geq 3$ implies $x\in\{a,b\}$ 
		\item $\mu_x \geq 1$ implies $\mu_{x\pm 1} = 0$
		\item $\mu_x{\geq 2}$ implies $\mu_{x\pm 2} = 0$ and $\mu_{x\pm 3}\leq 1$,
	\end{enumerate}
	and we consider the conditions to be vacuously true for any site $x\pm k \in \bZ\setminus \Lambda$. Moreover, the tiling  $T\in\cT_\Lambda$ for which $ \mu = \sigma_\Lambda(T) $ is unique, i.e., $  \sigma_\Lambda $ is injective.
\end{lem}
%
%Note that in the proof of Lemma~\ref{lem:tiling_properties}.1 we did not need to argue that the placed tiles did not overlap with one another as assuming $\mu\in\ran(\sigma_\Lambda)$ meant it was only necessary to show uniqueness of the tiling. While the proof for the reverse direction of Lemma~\ref{lem:tiling_configs} follows a similar strategy as that used in Lemma~\ref{lem:tiling_properties}.1, here one needs to show existence of a tiling, for which it is necessary to argue that any two placed tiles do not overlap.
%
\begin{proof}
	Given the set of tiles defined above and the boundary constraints, it is clear that any $\mu\in \ran(\sigma_\Lambda) $ satisfies Conditions 1-3.
	Conversely, suppose that $\mu \in \bN_0^\Lambda$ satisfies Conditions 1-3. We first determine the unique choice of boundary tiles (if any), and then place each type of bulk tile systematically from longest to shortest.
	
	If either $\mu_a\geq 2$ or $\mu_b\geq 2$, then combining Conditions 2-3, it is clear there are enough empty sites next to the boundary site to lay the corresponding tile $B_n^\#$, with $\#\in\{l,r\}$. Similarly, if $\mu_{b-1} = 2$, there are enough empty sites to lie $D^{(1)}$ and one can always place $M^{(1)}$ if $\mu_b = 1$. Moreover, such configurations cannot be covered by bulk tiles, and so this uniquely places the boundary tiles.
	
	For any remaining uncovered site $x$ for which $\mu_x = 2$, a bulk dimer must and can be placed to cover $x$ as $\mu_{x\pm1}=\mu_{x+2}=0$ by Conditions 2-3. These tiles do not overlap with one another or with any boundary tile as for any other $y\in\Lambda$ with $\mu_y \geq 2$, Conditions 2-3 imply $|x-y|\geq 3$ which is the minimum distance one needs to place two successive dimer tiles, or a dimer neighboring a boundary tile with two or more particles.
	
	Similarly, for any remaining uncovered site with $\mu_x = 1$, we must and can place a bulk monomer as Condition 2 guarantees that $\mu_{x+1}=0$. Once again, this tile does not overlap with any other previously placed tiles since Condition 2 guarantees this does not overlap with a neighboring monomer, and Condition 3 guarantees this does not overlap with any neighboring tile with two or more particles.
	
	All remaining uncovered sites hold no particles and thus must be tiled with voids. This completes the unique tiling $T$ that produces the configuration, i.e. $\mu = \sigma_\Lambda(T)$ as desired. 
\end{proof}

With respect to each root-tiling $R\in\cR_\Lambda$, we define \emph{the BVMD-subspace associated to $R$} and \emph{the space of all BVMD-tilings} by
\be\label{tiling_spaces}
\caC_\Lambda(R) = \spa \left\{ \ket{\sigma_\Lambda(T)}\, | \, T\in \cT_\Lambda(R)\right\}, \quad \caC_\Lambda = \overline{\spa} \{ \ket{\sigma_\Lambda(T)} \, | \, T \in \cT_\Lambda\},
\ee
respectively. Each $\caC_\Lambda(R)$ is finite-dimensional as there are only finitely many tilings $T$ connected to a root $R$. However, $\dim(\caC_\Lambda) = \infty$ as there are an infinite number of non-$\bZ$-induced boundary tiles. The following lemma summarizes some of the most important properties of these subspaces.
\begin{lem}[BVMD Tiling Space Properties] \label{lem:tiling_properties} Let $\Lambda=[a,b]$ be any interval with $|\Lambda|\geq 4$. 
	\begin{enumerate}
		%\item The map $\sigma_\Lambda$ is injective.
		\item $\caC_\Lambda(R)\perp \caC_\Lambda(R')$ for any pair of root-tilings $R\neq R'$. As a consequence,
		$
		\caC_\Lambda = \bigoplus_{R\in\cR_\Lambda}\caC_\Lambda(R).
		$
		\item $\caC_\Lambda(R) $ is an invariant subspace of $H_\Lambda$ for each $R\in\cR_\Lambda$ and the restriction of $H_\Lambda $ to $ {\caC_\Lambda} \subset \dom(H_\Lambda) $ is a bounded operator. In particular, $\caC_\Lambda$ is also invariant. % As a consequence, $H_\Lambda\caC_\Lambda \subseteq \caC_\Lambda$.
	\end{enumerate}
\end{lem}

\begin{proof}
	%	1. Suppose that $\mu = \sigma_{\Lambda}(T)$ for some tiling $T\in\cT_\Lambda$, and let $T=(T_1, \ldots, T_n)$ denoted the ordered tiling. We first determine if $T$ contains any boundary tiles. If $\mu_{a}\geq 2$, then $T_1 = B_{\mu_{a}}^l$. For the right boundary:
	%	\begin{itemize}
	%		\item If $\mu_{b}\geq 2$, then $T_n=B_{\mu_{b}}^r$.
	%		\item If $\mu_{b}$ = 1, then $T_n=M_r$.
	%		\item If $\mu_{b-1}=2$ and $\mu_{b}=0$, then $T_n=D_r$.
	%	\end{itemize}
	%All other (undefined) tiles in $T$ must be bulk tiles. Consider any site $x$ covered by a bulk tile in $T$ with $\mu_x>0$. If $\mu_x=2$, resp. $\mu_x =1$, then $T_{k_x}$ must be a (bulk) dimer, resp. monomer, where $k_x\in[1,n]$ is the index of the tile that covers $x$. Since these tiles cover all remaining particles in $\mu$, each remaining bulk tile in $T$ must be voids.
	
	1. The set of configuration states constitutes an orthonormal basis for $\cH_\Lambda$. Since $\cT_\Lambda(R) \cap \cT_\Lambda(R') = \emptyset$ for any two distinct roots $R\neq R'$, the  first result is an immediate consequence of the injectivity of $\sigma_\Lambda$ and the definition of $\caC_\Lambda(R)$, see \eqref{tiling_spaces}. The decomposition of $\caC_\Lambda$ is an immediate consequence of \eqref{tiling_spaces} since each $\caC_\Lambda(R)$ is finite-dimensional and the direct sum of countably many orthogonal closed subspaces is closed.
	
	2. It is trivial that $\caC_\Lambda(R)\subseteq\dom(H_\Lambda)$ as it is a span of a finite set of vectors in $\dom(H_\Lambda)$. We first show that $q_x^*q_x \ket{\sigma_\Lambda(T)} \in \caC_\Lambda(R)$ for each $T\in\cT_\Lambda(R)$ and $x\in[a+1,b-1]$. By direct computation, one finds
	\be\label{zero_configs}
	q_x^*q_x \ket{\sigma_\Lambda(T)} = 0 \in \caC_\Lambda(R)
	\ee
	if on the interval $[x-1,x+1] \subset \Lambda$, the configuration $\sigma_\Lambda(T)$ either has one particle on site $x$ and no particles at $x\pm1$, or $\sigma_\Lambda(T)$ has a pair of neighboring sites with no particles. One is thus left to consider tilings for which the particle configuration on $[x-1,x+1]$ is $ (101) $ or $ (020) $, that is, tilings $T^{M}\in \cT_\Lambda(R)$ with two consecutive monomers with particles at $x\pm1$, or tilings $T^D\in \cT_\Lambda(R)$ with a dimer ($D$ or $D^{(1)}$) with two particles at $x$. Note that these two sets are in one-to-one correspondence via a single replacement connecting $T^M\leftrightarrow T^D$.
	
Fixing a pair $T^M \leftrightarrow T^D$ as above, a direct computation yields
	\begin{align}
		q_x^*q_x \ket{\sigma_\Lambda(T^M)} & = |\lambda|^2 \ket{\sigma_{\Lambda}(T^M)} - \lambda\sqrt{2} \ket{\sigma_\Lambda(T^D)} \label{q_action_1}\\
		q_x^*q_x \ket{\sigma_\Lambda(T^D)} & = - \overline{\lambda}\sqrt{2} \ket{\sigma_\Lambda(T^M)}+ 2\ket{\sigma_{\Lambda}(T^D)} .\label{q_action_2}
	\end{align}
	Thus, the action of $q_x^*q_x$ on either kind of configuration produces a vector in $\caC_\Lambda(R)$ and $H_\Lambda \caC_\Lambda(R) \subseteq \caC_\Lambda(R)$ as claimed.
	From \eqref{zero_configs}-\eqref{q_action_2}, it  also follows that
	\[
	\|q_x^*q_x\|_{\caC_\Lambda(R)} := \sup_{0\neq \psi\in \caC_\Lambda(R)}\frac{\|q_x^*q_x\psi\|}{\|\psi\|} \leq |\lambda|^2+2
	\]
	where $|\lambda|^2+2$ is the largest eigenvalue of the $ 2\times 2 $ matrix
	\be\label{MM_action}
	\begin{bmatrix}
		|\lambda|^2 & -\overline{\lambda}\sqrt{2}\\
		-\lambda\sqrt{2} & 2
	\end{bmatrix}.
	\ee
	Therefore, $\|H_\Lambda\|_{\caC_\Lambda(R)} \leq (|\Lambda|-2)( |\lambda|^2+2)$. Since $R$ is arbitrary, the same bound holds for $H_\Lambda \restriction_{\caC_\Lambda}$ by part 1. Thus, $\caC_\Lambda \subseteq \mathrm{dom}(H_\Lambda)$ and the claimed invariance and boundedness holds. 
\end{proof}

\subsection{The ground state space for open boundary conditions}

We now turn to determining the ground states of $H_\Lambda$ on any interval $\Lambda$ with $|\Lambda|\geq 5$. We begin by proving that the ground-state space is contained in $\caC_\Lambda$, and then use this in combination with Lemma~\ref{lem:tiling_properties} to establish an orthogonal basis for the ground state space in Theorem~\ref{thm:gss}.

\begin{lem}[Support of Ground States]\label{lem:support}
For any interval $\Lambda =[a,b]$ with $|\Lambda|\geq 5$, the ground state space of $H_\Lambda$ is supported on BVMD-tilings, that is, $\cG_{\Lambda} \subseteq \caC_\Lambda$.
\end{lem}

\begin{proof}
Consider the expansion $ 
\psi = \sum_{\mu\in \bN_0^\Lambda} \psi(\mu) \ket{\mu} $ of an arbitrary ground state $\psi \in \caG_\Lambda$ in terms of the configuration basis. 
We use Lemma~\ref{lem:tiling_configs} and the frustration free property to show that $\psi(\mu)\neq 0$ implies $\mu\in \ran(\sigma_\Lambda)$. 

First, frustration-freeness guarantees that $\psi$ is in the kernel of each electrostatic interaction term $n_xn_{x+1}$. As such, for each $ \mu\in \bN_0^\Lambda $
\be\label{electrostatic_ff}
0 = \mu_x \mu_{x+1}\psi(\mu) \quad \mbox{for all} \;  x\in[a,b-1]  ,
\ee
and so $\mu$ satisfies Condition 2 of Lemma~\ref{lem:tiling_configs} if $\psi(\mu) \neq 0$. 

%We now consider the dipole-hopping interactions. Fix any $x\in[a+1,b-1]$ and configuration $\nu \in \bN_0^\Lambda$. Then acting on $\ket{\nu}$ by $q_x^*$, one finds
%\be\label{creation-action}
%q_x^* \ket{\nu} = \sqrt{\mu_x(\mu_x+1)}\ket{\mu} - \overline{\lambda}\sqrt{(\mu_{x-1}+1)(\mu_{x+1}+1)} \ket{\eta},
%\ee
%where we denote by $\mu$ and $\eta$, the two configurations on $\Lambda$ defined as:
%\begin{equation}\label{config_def}
%\mu_y= \begin{cases}
%\nu_y, & y \neq x \\
%\nu_y+2 & y = x
%\end{cases}, 
%\qquad 
%\eta_y= \begin{cases}
%\nu_y, & y \neq x\pm 1 \\
%\nu_y+1 & y = x\pm1
%\end{cases}
%\end{equation}
Second, frustration-freeness also implies $\psi \in \ker(q_x)$ for any $x\in[a+1,b-1]$. In particular, $0 = (q_x\psi)(\nu)$ for all $ \nu  \in \bN_0^\Lambda $ from which it follows that $ \psi(\mu) \neq 0 $ if and only if $\psi(\eta) \neq 0$ where $\mu$ and $\eta$ are the two associated configurations (see \eqref{creation-action}):
\begin{equation}\label{config_def}
 \mu := (\alpha^*_x)^2\nu , \qquad \eta := \alpha^*_{x+1}  \alpha^*_{x-1} \nu .
 \end{equation}
If there is $x\in[a+1,b-1]$ such that $\mu_x \geq 3$, then considering \eqref{config_def} the configuration $\eta$ associated to $\nu = \alpha_x^2\mu$ satisfies $\eta_x\eta_{x\pm1} >0$, and hence $ \psi(\mu) = \psi(\eta) = 0 $ by \eqref{electrostatic_ff}. Therefore, Condition 1 of Lemma~\ref{lem:tiling_configs} holds if $\psi(\mu)\neq 0$.

Now, consider any configuration $\eta\in \bN_0^\Lambda$ for which $\eta_{x-1} \geq 2$ and $\eta_{x + 1}>0$ for some $x\in[a+1,b-1]$. Then the configuration $\nu=\alpha_{x-1}\alpha_{x+1}\eta$ is well-defined, and the configuration $\mu$ as in \eqref{config_def} satisfies $\mu_{x-1}\mu_x >0$. Arguing as in the previous case we again find $\psi(\eta) = \psi(\mu) = 0$. The analogous argument holds if $\eta_{x+1}\geq 2$ and $\eta_{x-1}>0$. Therefore, if $\psi(\eta) \neq 0$ and $\eta_x \geq 2$ for some $x\in\Lambda$, then $\eta_{x\pm 2} = 0$.

To show that $\psi(\mu) \neq 0$ implies Condition 3 of Lemma~\ref{lem:tiling_configs} for $ \mu $, it is only left to show that  $\psi(\mu) = 0$ if $\min\{\mu_x,\,\mu_{x+3}\} \geq 2$ for some $x\in[a,b-3]$. Since $|\Lambda|\geq 5$, it is clear that either $x$ or $x-3$ is an interior site. Assume that $x>a$, and define $\nu=\alpha_x^2\mu$. Then, $\eta$ as in \eqref{config_def} satisfies $\eta_{x+1} > 0$ and $\eta_{x+3} \geq 2$. By the previous case this implies $0=\psi(\eta) = \psi(\mu)$. The analogous argument holds in the case that $x+3$ is interior, where we apply \eqref{config_def} with $\nu=\alpha_{x+3}^2\mu$ and $\eta=\alpha_{x+2}^*\alpha_{x+4}^*\nu$. This completes the proof.
\end{proof}

To summarize, the results up to this point, we have found that every BVMD-tiling space $\caC_\Lambda(R)$ is a closed invariant subspace of the Hamiltonian $H_\Lambda$ and any two distinct BVMD-spaces are orthogonal. Moreover, for $|\Lambda|\geq 5$ the ground state space is contained in the closed span of all BVMD-tilings $\caC_\Lambda$. Since $ \caG_\Lambda \subset \caC_\Lambda = \bigoplus_{R\in\cR_\Lambda} \caC_\Lambda(R)$, the orthogonality and invariance of the individual BVMD-spaces  imply
\be\label{gs_direct_sum}
\caG_\Lambda = \bigoplus_{R\in\cR_\Lambda} (\caG_\Lambda \cap \caC_\Lambda(R)).
\ee
Hence, one can build a orthogonal basis for $\caG_\Lambda$ by finding an orthogonal basis of each $\caG_\Lambda \cap \caC_\Lambda(R)$ and taking the union over all root tilings. We prove in Theorem~\ref{thm:gss} that each $\caG_\Lambda \cap \caC_\Lambda(R)$ is one-dimensional and spanned by the \emph{BVMD-state} $\psi_\Lambda(R)$ defined by
\begin{equation}
\label{BVMD}
\psi_\Lambda(R) := \sum_{T\in \cT_\Lambda(R)} \left(\frac{\lambda}{\sqrt{2}}\right)^{d(T)} \ket{\sigma_{\Lambda}(T)}
\end{equation}
where $d(T)$ is the number of dimers $D$ or $D^{(1)}$ in the tiling $T$. Our convention implies that $d(R) = 0$ for all root tilings.

\begin{theorem}[OBC Ground-State Space]\label{thm:gss} Fix an interval $\Lambda $ with $|\Lambda|\geq 5$. For any root-tiling $R\in\cR_\Lambda$, one has
	\be\label{BVMD_GS}
	\caG_\Lambda \cap \caC_\Lambda(R) =\spa\{\psi_\Lambda(R)\}.
	\ee
Thus, the BVMD-states form an orthogonal basis of the ground state space $\cG_\Lambda$,
	\be\label{G_Lambda}
	\caG_\Lambda = \spa \{ \psi_\Lambda(R) \, | \, R\in \cR_\Lambda\}.
	\ee
\end{theorem}

\begin{proof}
	Any vector 
	\[
	  \psi = \sum_{T\in\cT_\Lambda(R)} c_T\   \ket{\sigma_\Lambda(T)} \in  \caC_\Lambda(R)  \quad\mbox{with coefficients $ c_T:= \psi(\sigma_\Lambda(T)) $} 
	\]
	 is in $ \caG_\Lambda $ if and only if $\psi \in\ker(q_x)\cap \caC_\Lambda(R)$ for all interior sites $x\in[a+1,b-1]$. 
	Using the criterion from Lemma~\ref{lem:tiling_configs} it is easy to check that $q_x \ket{\sigma_\Lambda(T)} = 0$ for all tilings $T$ except those that have either a pair of neighboring monomers with particles at $x\pm 1$, or a dimer with two particles at $x$. Consequently, if $ \psi \in\ker(q_x)\cap \caC_\Lambda(R) $ then 
	\be\label{qx_kernel}
	0=q_x\psi = \sum_{T^M\in\cT_\Lambda^M(R)} q_x\left(c_{T^M}\ket{\sigma_\Lambda(T^M)}+c_{T^D}\ket{\sigma_\Lambda(T^D)}\right),
	\ee
	where $\cT_\Lambda^M(R)$ denotes the set of tilings of $\Lambda$ that have two monomers with particles at $x\pm 1$, and $T^D$ is the tiling obtained by replacing these two monomers with a dimer in $T^M$. A direct computation shows that
	\be\label{qx-action}
	q_x(c_{T^M}\ket{\sigma_\Lambda(T^M)}+c_{T^D}\ket{\sigma_\Lambda(T^D)}) = (-\lambda c_{T^M}+\sqrt{2}c_{T^D})\ket{\sigma_\Lambda(T^V)}
	\ee
	where $T^V$ is the tiling obtained by replacing the two monomers at $x\pm 1$ with voids. Noting that $T^V \neq \tilde{T}^V$ for any pair of distinct $T^M, \tilde{T}^M\in \cT_\Lambda^M(R)$, combining \eqref{qx_kernel} with \eqref{qx-action} implies that
	\be\label{gs_coeff}
	c_{T^D} = \frac{\lambda}{\sqrt{2}} c_{T^M}.
	\ee
	
	Conversely, given any pair of tilings $T^M,T^D\in\cT_\Lambda(R)$ that differ only by a single replacement of two monomers by a dimer, there is an interior $x\in [a+1,b-1]$ for which \eqref{qx-action} holds and, hence, the respective coefficients satisfy \eqref{gs_coeff}. By definition, every $T\in\cT_\Lambda(R)$ can be connected to the root tiling $R$ by replacing all dimers $D$ or $D^{(1)}$ by a pair of neighboring monomers. Thus, inductively applying \eqref{gs_coeff} shows
	\[
	c_T = c_R\left(\frac{\lambda}{\sqrt{2}}\right)^{d(T)}\quad \text{for all}\quad T\in\cT_\Lambda(R),
	\]
	from which it follows that $\psi = c_R \psi_\Lambda(R)$. This completes the proof.	
\end{proof}

\subsection{Properties of BVMD States}\label{subsec:BVMD_properties}
We briefly summarize some important properties of BVMD-states, the proofs of which are immediate consequences of the previous results, or simple modifications of the equivalent statements found in \cite{NWY:2021}.

\begin{enumerate}
\item 
Applying the replacement rules \eqref{eq:replacement} to any tiling $T\in\cT_\Lambda$ leaves the number of particles invariant. As a consequence, each BVMD-state is an eigenstate of the number operator $N_\Lambda= \sum_{x\in \Lambda} n_x $, 
\[
N_\Lambda\psi_\Lambda(R) = \sum_{x\in \Lambda} \sigma_\Lambda(R)_x \  \psi_\Lambda(R) . 
\]
Moreover, the orthogonality of distinct BVMD-spaces immediately implies that
\[
\braket{\psi_\Lambda(R')}{\psi_\Lambda(R)} = \delta_{R,R'} \sum_{T\in \cT_\Lambda(R)} \left(\frac{|\lambda|^2}{2}\right)^{d(T)}
\]
and $\dim(\caG_\Lambda) = |\cR_\Lambda|=\infty$, as there are an infinite number of non-$\bZ $-induced boundary tiles.

\item Observing that voids are unaffected by the replacement rules, each BVMD-state can be factored (up to possible boundary states) using void states $\ket{0}$, and \emph{squeezed Tao-Thouless states} $\vp_{L+1}^{(i)}\in \cH_{[1,2L+i]}$. For fixed $L\geq 0$ and $i\in\{1,2\}$, the squeezed Tao-Thouless state $\vp_{L+1}^{(i)}$ is the BVMD-state generated by the root tiling that covers $2L+i$ sites with monomers, that is
\be\label{TT}
\vp_{L+1}^{(i)} := \psi_{[1,2L+i]}(M_{L+1}^{(i)}), \quad M_{L+1}^{(i)} = (M, M, \ldots, M, \, M^{(i)}),
\ee
where $M^{(2)}=M$, and $M_{L+1}^{(i)}$ has $L+1$ tiles, see Figure~\ref{fig:M3}. We will also write $\vp_L:=\vp_L^{(2)}$ and use the convention $\vp_0 = 1$.
\begin{figure}
	\begin{center}
		\includegraphics[scale=.3]{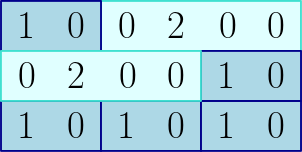}
	\end{center}
\caption{The tilings generated from $M_3^{(2)}$.} % from which recursion relation $\vp_3=\vp_2\otimes\ket{10}+\tfrac{\lambda}{\sqrt{2}}\vp_1\otimes\ket{0200}$ is clear.}
\label{fig:M3}
\end{figure}

To factorize an arbitrary BVMD-state $\psi_\Lambda(R)$, let $ \{v_1, \ldots, v_k\}\subseteq \Lambda =[a,b]$ be the ordered set of sites covered by voids in the root tiling $R=(R_1,\ldots, R_m)$, and denote by $L_i\in\bN_0$, $i=1, \ldots, k+1$, the number of monomers ($M$ or $M^{(1)}$) between $v_{i-1}$ and $v_i$. Here, we use the convention that $v_0 = a-1$ and $v_{k+1} = b+1$. Then, $\psi_\Lambda(R)$ factors as
\be\label{fragmentation}
\psi_\Lambda(R) = \psi^l \otimes \vp_{L_1}\otimes\ket{0}_{v_1}\otimes \ldots \otimes\vp_{L_{k}}\otimes\ket{0}_{v_k}\otimes\psi^r
\ee
where the boundary states $\psi_{l}$, $\psi_r$ are: 
\be
\psi^l = \begin{cases}
	\ket{n00} & \text{if}\;\;R_1 =B_n^l\\
	1 & \text{otherwise}
\end{cases}
 \qquad
\psi^r = \begin{cases}
	\vp_{L_{k+1}}\otimes \ket{0n}, & \text{if}\;\;R_m =B_n^r \\
	\ket{\vp_{L_{k+1}}^{(1)} }& \text{if}\;\;R_m = M^{(1)}\\
	\ket{\vp_{L_{k+1}}} & \text{otherwise}
\end{cases},
\ee
see Figure~\ref{fig:equivalence_class}. The formal proof of this expression follows from a slight modification the argument used in \cite[Theorem 2.10]{NWY:2021}.
\item 
As a fundamental building block of the BVMD-states, the squeezed Tao-Thouless states and their properties play a key role in our analysis. Since the bulk monomer and dimer both end in a vacant site, for each $L\geq 1$
\be\label{vp2_to_vp1}
\vp_{L} = \vp_L^{(1)} \otimes \ket{0}.
\ee

In view of the substitution rules \eqref{eq:replacement}, for either $i\in\{1,2\}$ these states can be further decomposed according to the following \emph{recursion relations}: for any $n=l+r$ with $l\geq 1$ and $r\geq 2$, 
\be\label{recursion_general}
\vp_n^{(i)} = \vp_l\otimes \vp_r^{(i)} +\frac{\lambda}{\sqrt{2}} \vp_{l-1} \otimes \ket{\sigma_d}\otimes\vp_{r-1}^{(i)}
\ee
where $\ket{\sigma_d} = \ket{0200}$. In the case that $r=1$, one also has the modified relation
\be\label{recursion}
\vp_n^{(i)} = \vp_{n-1}\otimes\vp_1^{(i)} +  \frac{\lambda}{\sqrt{2}} \vp_{n-2} \otimes \ket{\sigma_d^{(i)}}
\ee
where $\ket{\sigma_d^{(1)}} := \ket{020}$ and $\ket{\sigma_d^{(2)}} := \ket{\sigma_d}$, see Figure~\ref{fig:M3}.

\item The final property is an expression for the ratio $\beta_n:=\|\vp_{n-1}\|^2/\|\vp_{n}\|^2$ and follows from observing that the two vectors on the right side of \eqref{recursion} are orthogonal. As such $\|\vp_n^{(i)}\|^2 = \|\vp_n\|^2$ for all $i$ and
\be\label{norm_recursion}
\|\vp_n\|^2 = \|\vp_{n-1}\|^2 + \frac{|\lambda|^2}{2}\|\vp_{n-2}\|^2.
\ee
By applying the argument of \cite[Lemma 2.13]{NWY:2021}, this relation indicates that the ratio $\beta_n$ converges as $n\to \infty$. Specifically,
\be\label{alpha_convergence}
\beta_n = \frac{1}{\beta_+}\cdot\frac{1-\beta^n}{1-\beta^{n+1}} \to \frac{1}{\beta_+}
\ee
where $\beta = \frac{\beta_-}{\beta_+}\in(-1,0)$ and $\beta_{\pm} = (1\pm \sqrt{1+2|\lambda|^2})/2$.
\end{enumerate}

\subsection{Tiling spaces and ground states for periodic boundary conditions}\label{sec:periodic_tilings}

For the ground state of  $H_\Lambda^\per$ the relation in \eqref{dipole_gs} holds at every site in $\Lambda=[a,b]$. Hence, the ground state space can be described in terms of tilings  that only require the bulk tiles $V$, $M$, and $D$. As defined at the beginning of Section~\ref{sec:VMD}, we call any cover $T$ of the ring $\Lambda$ by these tiles a \emph{periodic VMD-tiling}, and further say it is a \emph{periodic root tiling} if it only consists of bulk monomers and voids. 
Any periodic tiling can be written in a (non-unique) ordered form $T = (T_1, \ldots, T_k)$ as long as the location of the first tile, e.g. the one covering $ a $,  is specified. Two periodic tilings are then called \emph{connected}, denoted $T\leftrightarrow T'$, if they can be transformed into one another using the bidirectional replacement rule $(10)(10)\leftrightarrow (0200)$, for which  we consider the first and last tiles in $T$ to be neighbors. The set of periodic root tilings $\cR_{\Lambda}^\per$ partitions the set of all periodic tilings $ \cT_\Lambda^\textrm{per} $  via this equivalence relation. An invariant subspace of the Hamiltonian $H_\Lambda^\per$ is given by 
\[
\caC_\Lambda^\per(R) := \spa\left\{\ket{\sigma_\Lambda(T)} | T\leftrightarrow R\right\}
\]
where $\sigma_\Lambda(T)\in\bN_0^\Lambda$ is again the particle configuration associated with the periodic tiling $T\in  \cT_\Lambda^\textrm{per} $, cf.~\eqref{config}.
A consequence of Lemma~\ref{lem:periodic_configs} below is that these tiling spaces are again mutually orthogonal and
\be\label{per_tiling_space}
\caC_\Lambda^\per := \spa \left\{\ket{\sigma_\Lambda(T)} | \, T\in\cT_{\Lambda}^\per\right\} = \bigoplus_{R\in\cR_{\Lambda}^{\per}}\caC_\Lambda^\per(R).
\ee
This subspace will turn out to be finite-dimensional, and hence closed.  

Note that cutting a periodic tiling between the endpoints $a$ and $b$ produces a BVMD tiling of the interval $\Lambda$, and so one can identify $\cT_\Lambda^\per \subseteq \cT_\Lambda$. As such, configurations that arise from period tilings can be characterized in a similar, in fact, even simpler way than done in Lemma~\ref{lem:tiling_properties}.
\begin{lem}[Periodic VMD-Tiling Configurations]\label{lem:periodic_configs}
	Given a ring $\Lambda = [a,b]$ with $|\Lambda|\geq 4$, a configuration $\mu\in\bN_0^{\Lambda}$ is in the range of the restriction $ \sigma_\Lambda: \cT_\Lambda^\textrm{per}  \to \bN_0^{\Lambda} $ if and only if the following two conditions hold:
	\begin{enumerate}
		\item $\mu_x \geq 1$ implies $\mu_{x\pm 1} = 0$
		\item $\mu_x{\geq 2}$ implies $\mu_{x\pm 2} = 0$ and $\mu_{x\pm 3}\leq 1$,
	\end{enumerate}
	where $x\pm k$ is taken modulo $|\Lambda|$. Moreover, the tiling  $T\in \cT_\Lambda^\textrm{per} $ for which $ \mu = \sigma_\Lambda(T) $ is unique, i.e., $  \sigma_\Lambda \restriction_{\cT_\Lambda^\textrm{per}} $ is injective.
\end{lem}

The proof of this result follows exactly as that of Lemma~\ref{lem:tiling_configs} without the case of boundary tiles and with the observation that any tiling configuration $\sigma_\Lambda(T)$ with $T\in \cT_\Lambda^\textrm{per} $  has at most two particles at any site.  Using this result, we establish the following properties of the ground state space.

\begin{theorem}[Periodic Ground State Space] \label{thm:periodic_gss} The following properties hold for the ground state space $\caG_\Lambda^\per$ on any ring $\Lambda = [a,b]$ with $|\Lambda|\geq 4$:
	\begin{enumerate}
		\item The set of periodic VMD-states $\{\psi_\Lambda^\per(R) | R\in \cR_\Lambda^\per\}$ is an orthogonal basis of $\caG_\Lambda^\per$ 
		where
		\be
		\psi_\Lambda^{\per}(R) = \sum_{T\in \cT_\Lambda^\per(R)} \left(\frac{\lambda}{\sqrt{2}}\right)^{d(T)}\ket{\sigma_\Lambda(R)}
		\ee
		and $d(T)$ is again the number of dimers $D$ in the periodic tiling $T$. 
		\item The dimension grows exponentially in the system size. Specifically, there are positive constants $c,C>0$ independent of $\Lambda$ for which
		\be\label{gs_dimension}
		c\mu_+^{|\Lambda|} \leq \dim \cG_\Lambda^\per \leq C\mu_+^{|\Lambda|},
		\ee
		where $\mu_+ := (1+\sqrt{5})/2$.
		\item For any periodic root tiling, $N_\Lambda\psi_\Lambda^\per(R)=N_\Lambda(R)\psi_\Lambda^\per(R)$, where $N_\Lambda(R)$ is the number of particles in $R\in \cR_\Lambda^\per$. Moreover, the ground state is at most half filled, 
		\be\label{gs_filling}
		\frac{1}{2}-\frac{1}{2|\Lambda|} \leq \max_{R\in\cR_\Lambda^\per}\frac{N_\Lambda(R)}{|\Lambda|} \leq \frac{1}{2}.
		\ee
	\end{enumerate}
\end{theorem}

\begin{proof}
	1. This result follows from the same argument used in the proof of Theorem~\ref{thm:gss}. %\vspace{12pt}
	
	2. From part 1, it is clear that $\dim\cG_\Lambda^\per = |\cR_\Lambda^\per|$. 
	Any periodic root tiling of $\Lambda = [a,b]$ considered as a ring either covers $\{a,b\}$ with a monomer, or is a root tiling of the interval $[a,b]$ by monomers and voids. As such,
\[
|\cR_\Lambda^\per| = r_{|\Lambda|} + r_{|\Lambda|-2}
\]
where $r_{|\Lambda|}$ is the number of tilings of an interval of length $|\Lambda|$ with monomers and voids, which is clearly finite. 
	Thus, we need only count the number of tilings $r_{L}$ that cover an interval of size $L$ (with open boundary conditions) by voids and monomers. This number satisfies the recursion relation
	\be\label{tile_count}
	r_L = r_{L-1} + r_{L-2}
	\ee
	with initial conditions $r_1=1$ and $r_2 =2$. The solution reads $r_L = (\mu_+^{L+1}-\mu_-^{L+1})/\sqrt{5}$ where $\mu_{\pm} = (1\pm\sqrt{5})/2$, from which the result follows. %\vspace{12pt}
	
	3. The claim $N_\Lambda\psi_\Lambda^\per(R) = N_\Lambda(R)\psi_\Lambda^\per(R)$ is clear since the replacement rule does not change the total number of particles. The value $N_\Lambda(R)$ is maximized by any root-tiling $R_{\max}$ with $\lfloor\frac{|\Lambda|}{2}\rfloor$ monomers. This gives 
	\[
	N_\Lambda(R_{\max}) = \begin{cases} 
		|\Lambda|/2, & |\Lambda| \text{ even} \\ 
		(|\Lambda|-1)/2, & |\Lambda| \text{ odd} 
	\end{cases}
	\]
	which establishes \eqref{gs_filling}.
\end{proof}

To conclude this section, we briefly comment on the decay of ground state correlations. Theorem~\ref{thm:periodic_gss} establishes that an orthogonal basis for the ground state space is labeled by the periodic root tilings. Similar to \cite[Theorem 4.1]{NWY:2021}, we expect that each periodic VMD-state will exhibit exponential decay of correlations for bounded observables. In contrast, in \cite[Section 4.3]{NWY:2021} it was pointed out that due to the exponential degeneracy of the ground state, other pure ground states with arbitrarily slow decay of correlations could be constructed. We expect that similar examples can be created for the present $\nu=1/2$ model.

\section{Proof of a uniform gap in the BVMD tiling space}\label{sec:MM}
We now apply the martingale method to produce a  lower bound on the spectral gap $E_1(\caC_\Lambda)$ corresponding to open boundary conditions that is uniform in the volume. The martingale method can be used to estimate the spectral gap above the ground state of a frustration-free Hamiltonian on a finite-dimensional Hilbert space. While in previous works it has been used to study spectral gaps for finite-volume quantum spin and lattice fermion models, we adapt it here to the present lattice boson model.
\subsection{Reduction to a finite dimensional subspace}
 As remarked earlier, one difficulty in adapting the martingale method is that the Hilbert space $ \cH_\Lambda $ and the tiling subspace $ \caC_\Lambda $ are both infinite dimensional. For the present model, we solve this issue and establish an initial estimate on the finite-volume gap by observing that $ E_1(\caC_\Lambda)   $ is realized on the invariant subspace associated to \emph{bulk BVMD-tilings} of $ \Lambda $, which turns out to be finite dimensional. This set is the collection of all tilings generated by the substitution rules on a subset $\cR_\Lambda^{\infty} \subset\cR_\Lambda$ of root-tilings $R=(R_1, \ldots, R_k) $ for which the boundary tiles are restricted to
\be\label{na_bdy_conds}
R_1 \in \{V, \, M, \, B_2^l\}, \quad R_k \in \{V, \, M, \, M^{(1)}, \, B_2^r\} .
\ee
Said differently, this is precisely the set of tilings obtained from truncating VMD-tilings of $\bZ$. 
The corresponding subspace of \emph{$\bZ $-induced  BVMD-tilings}, or \emph{bulk tilings} for short, is abbreviated by
\be\label{def:IVBVMD}
\caC_\Lambda^{\infty} := \bigoplus_{R\in\cR_{\Lambda}^{\infty}}\caC_\Lambda(R).
\ee
 Since each subspace $ \caC_\Lambda(R) $ is invariant for $ H_\Lambda $, so too is  $ \caC_\Lambda^{\infty}  \subseteq\dom(H_\Lambda)$, which allows us to define the gap
 \be\label{eq:gapinvariant}
 E_1(\caC_\Lambda^\infty) := \inf_{0\neq \psi\in\caC_\Lambda^\infty \cap \caG_\Lambda^\perp}\frac{\braket{\psi}{H_\Lambda\psi}}{\|\psi\|^2}.
 \ee
 \begin{theorem}[Restriction to Bulk Tilings]\label{thm:gap_reduction} For any interval $\Lambda = [a,b]$:
 \begin{enumerate}\itemsep0pt
 \item $ \dim \caC_\Lambda^{\infty} < \infty $,
 \item $ \displaystyle 
E_1(\caC_\Lambda) \geq  \min\left\{ E_1(\caC_\Lambda^{\infty}) , \, E_1(\caC_{[a+3,b]}^{\infty}), E_1(\caC_{[a,b-2]}^{\infty}), E_1(\caC_{[a+3,b-2]}^{\infty}) \right\} $ is strictly positive, where we use the convention that $E_1(\mathcal{V}_{\Lambda'})=\infty$ if $\mathcal{V}_{\Lambda'} \subseteq \caG_{\Lambda'}$ or $ \Lambda' = \emptyset $. 
%\item $N_\Lambda\psi = N_\Lambda(R)\psi$ for any state $\psi\in\caC_\Lambda(R)$ and $R\in\cR_\Lambda$  where $N_\Lambda(R)$ is the number of particles in the root tiling. Moreover, the $\bZ$-induced tilings satisfy
%\be\label{filling}
%\frac{1}{2}\leq \max_{R\in\cR_\Lambda^\infty}\frac{N_\Lambda(R)}{|\Lambda|} \leq \frac{1}{2}+\frac{3}{2|\Lambda|}
%\ee
\end{enumerate}
\end{theorem}
 \begin{proof}
1. It suffices to show that $|\cR_\Lambda^{\infty}|<\infty$ since $\dim(\caC_\Lambda(R))<\infty$ for each $R\in \cR_\Lambda$. The number of root tilings $R\in\cR_\Lambda^{\infty}$ that cover $\Lambda$ with just voids $V$ and bulk monomers $M$ satisfies the recursion relation from \eqref{tile_count}. As a consequence, the number of root tilings with a fixed pair of boundary tiles $R_l$, $R_r$ is given by $r_{L-|R_l|-|R_r|}$ where $|R_\#|$ is the length of the tile $ R_\# $. Using the convention that $r_0 = 1$, this implies
\[
|\cR_\Lambda^{\infty}| = \sum_{\substack{R_l\in\{V,M,B_2^l\}\\R_r\in\{V,M,M^{(1)},B_2^r\}}}r_{|\Lambda|-|R_l|-|R_r|}
\]
which is clearly finite.
%Specifically, given the solution to the recursion relation, one concludes there are $c,C>0$ independent of $|\Lambda|$ such that
%$ 
%c\mu_+^{|\Lambda|} \leq |\cR_\Lambda^{\infty}| \leq C\mu_+^{|\Lambda|} $.

2. Since $\caC_\Lambda = \bigoplus_{R\in\cR_\Lambda}\caC_\Lambda(R)\subseteq\dom(H_\Lambda)$ is a sum of orthogonal, invariant subspaces all of which contain a unique ground state, $\psi_\Lambda(R)$, the spectral gap on $\caC_\Lambda $ is the infimum over the gaps in each subspace, 
\[
E_1(\caC_\Lambda) = \inf\{E_1(\caC_\Lambda(R)) \,| \, R\in \cR_\Lambda\} . 
\]
The analogous argument implies $E_1(\caC_\Lambda^\infty) = \inf\{E_1(\caC_\Lambda(R)) \,| \, R\in\cR_\Lambda^\infty\}$. Thus, we need only consider $E_1(\caC_\Lambda(R))$ for $R\in\cR_\Lambda\setminus\cR_\Lambda^\infty$. Suppose that $R=(R_1,\ldots,R_k)\in\cR_\Lambda$ is such that both boundary tiles do not belong to the sets in~\eqref{na_bdy_conds}, that is, $R_1 = B_n^l$ and $R_k = B_m^r$ for some $n,m\geq 3$. Since the replacement rules do not apply to these tiles, any nonzero $\psi_\Lambda \in \caC_\Lambda(R)$ factors as
\[
\psi_\Lambda = \ket{n00}\otimes \psi_{\Lambda'}\otimes \ket{0m}
\]
where $\psi_{\Lambda'}\in \caC_{\Lambda'}(R')$, $\Lambda'=[a+3,b-2]$, and $R'=(R_2,\ldots, R_{k-1})\in\cR_{\Lambda'}^{\infty}$. Moreover, \eqref{artificial_action} shows
\[
H_\Lambda\psi_\Lambda = \ket{n00}\otimes H_{\Lambda'}\psi_{\Lambda'}\otimes \ket{0m},
\]
and so $E_\Lambda^1(\caC_\Lambda(R)) = E_{\Lambda'}^1(\caC_{\Lambda'}(R'))$. A similar statement can be made in the case that only one, but not both, of the boundary tiles of $R$ do not belong to~\eqref{na_bdy_conds}. This proves the asserted inequality. The strict positivity of the spectral gap follows from the fact that $ H_{\Lambda} $ restricted to $ \caC_{\Lambda}^\infty $ for any finite interval $ \Lambda $ is unitarily equivalent to a matrix.
%
%3. The first claim is immediate since the replacement rules are particle number preserving. Recalling \eqref{TT}, the lower bound in \eqref{filling} is clear since $N_\Lambda(M_L^{(i)})\geq |\Lambda|/2$ where $|\Lambda| = 2(L-1)+i$. The interior of any root tiling is maximally filled when covered by monomers, and so for any $R=(R_1,\ldots,R_k)\in\cR_\Lambda^\infty$, 
%\[
%N_\Lambda(R) \leq \frac{|\Lambda|-|R_1|-|R_k|}{2} + n(R_1)+n(R_k)
%\] 
%where $|R_i|$ is the length of $R_i$ and $n(R_i)$ is the number of particles in $R_i$. The upper bound in \eqref{filling} follows from maximizing the above over all possible boundary conditions, which occurs at $R_1=B_2^l$ and $R_k = B_2^r$.
\end{proof}

\subsection{The martingale method}\label{subsec:martingale}
We are now able to apply the martingale method proved in \cite[Theorem 5.1]{nachtergaele:2016b} to produce a lower bound on the spectral gap $E_1(\caC_\Lambda^{\infty})$. Since our Hamiltonian is translation invariant, it is sufficient to consider $\Lambda = [1,L]$, and Theorem~\ref{thm:tiling_gap} below establishes that
$
\inf_{L\geq 7}E_1(\caC_{[1,L]}^{\infty}) >0 $. 

To state the main result in this section, recall that for any operator $A $ on $\cH_{\Lambda'} $, 
the mapping 
\be\label{idenfification} A \mapsto A\otimes \1_{\Lambda\setminus \Lambda'}
\ee
identifies $A$ as an operator on $\cH_\Lambda$ for any finite $\Lambda \supseteq \Lambda'$. We introduce several sequences of positive operators of this type associated to a fixed $\Lambda = [1,L]$ with $L \geq 7$. Let $N\geq 3$, $k\in\{1,2\}$ denote the unique integers so that $L=2N+k$ and define two finite sequences of Hamiltonians $H_n,h_n \geq 0$ for $2\leq n \leq N$ by
\be\label{ham_sequences}
H_n = \sum_{k=2}^n h_k, \quad
h_{n} = H_{\Lambda_n} \quad \text{where}\quad
\Lambda_n = \begin{cases}
	[1,4+k], & n=2 \\
	[2n+k-5,2n+k], & 3\leq n\leq N
\end{cases}
\ee
The associated sequence of intervals satisfies $|\Lambda_2|\in\{5,6\}$, $|\Lambda_n| = 6$ for $n\geq 3$, and $|\Lambda_n\cap\Lambda_{n+1}|=4$ for $2\leq n< N$, from which one can check that each interaction term supported on $\Lambda$ ($n_xn_{x+1}$ or $q_x^*q_x$) is a summand in at least one and most three of the Hamiltonians $h_n$. As a result, for all $2\leq n \leq N$,
\be\label{equivalent_hams}
H_{[1,2n+k]} \leq H_n \leq 3 H_{[1,2n+k]}
\ee
and, in particular, $H_\Lambda \leq H_{N} \leq 3H_{\Lambda}$. An important consequence of \eqref{equivalent_hams} is that the ground-state spaces agree. Thus, 
\be
\ker H_n = \cG_{[1,2n+k]}\otimes\cH_{[2n+k+1,L]}\subseteq \cH_\Lambda  
\ee
where $\cG_{[1,2n+k]}$ is as in Theorem~\ref{thm:gss}. Let $G_n$ denote the orthogonal projection onto $\ker H_n$. By frustration-freeness, $\ker H_{n+1}\subseteq \ker H_n$ for each $n$, and so the following resolution of the identity forms a mutually orthogonal family of orthogonal projections:
\begin{equation}\label{En}
	E_n := \begin{cases}
		\1 - G_2 & n = 1 \\
		G_n - G_{n+1} & 2\leq n \leq N-1 \\
		G_{N} & n=N
	\end{cases}
\end{equation}
Finally, we denote by $g_n,$ $2\leq n \leq N$, the orthogonal projection onto $\ker h_n=\cG_{\Lambda_n}\otimes\cH_{\Lambda\setminus\Lambda_n}\subseteq\cH_\Lambda$.

For our application of the martingale method, we consider the restriction of these operators to the subspace $\caC_\Lambda^{\infty}=\bigoplus_{R\in\cR_\Lambda^{\infty}}\caC_\Lambda(R) $. The BVMD-space $\caC_\Lambda(R)$  for each root-tiling $R\in\cR_\Lambda$ is invariant under $H_{\Lambda'}$ for any $\Lambda' \subseteq \Lambda$ as this is invariant under all $q_x^*q_x$ supported on $\Lambda$ (and in particular those supported on $\Lambda'$). By the same reasoning, $ \caC_\Lambda(R) $ is invariant under $h_n$ and $H_n$ as well as the associated ground-state projections $g_n$ and $G_n$ for all $n\geq 2$. Hence each of these self-adjoint operators can be jointly block-diagonalized with respect to the decomposition $\cH_\Lambda = \caC_\Lambda^{\infty}\oplus(\caC_\Lambda^{\infty})^\perp.$ Explicitly, for any $2\leq n \leq N$:
\begin{align}\label{blocks}
	A_n & = A^{\infty}_n + (\1 - P_{\caC_\Lambda^{\infty}}) A_n(\1 - P_{\caC_\Lambda^{\infty}}) , \qquad \mbox{with} \quad A^{\infty}_n:=A_n\restriction_{\caC_\Lambda^{\infty}} =P_{\caC_\Lambda^{\infty}} A_n P_{\caC_\Lambda^{\infty}}.
\end{align}
where $A_n \in \{H_n,\, h_n,\, G_n, \, g_n\}$ and $P_{\caC_\Lambda^{\infty}}$ is the orthogonal projection onto $\caC_{\Lambda}^{\infty}$. This block diagonalization also extends to every $E_n$ by \eqref{En}. Thus, the restriction of any of these operators to $\caC_\Lambda^{\infty}$ is given by the associated block diagonal component $ A^{\infty}_n $. 

We are  now able to state the main result in this section:

\begin{theorem}[Application of the Martingale Method]\label{thm:tiling_gap}
	Fix $\Lambda =[1,L]$ with $L\geq 7$, and let $h_n$, $g_n$ and $E_n$ be as in \eqref{ham_sequences} and \eqref{En}. The restrictions of these operators to $\caC_\Lambda^{\infty}$ satisfy the following three criterion:
	\begin{enumerate}
		\item For all $n\geq 2$, $h_n^{\infty}\geq 2\kappa(\1-g_n^{\infty})$.
		\item For all $n\geq 2$, $[g_n^{\infty},E_m^{\infty}]\neq 0$ only if $m \in [n-3,n-1]$.
		\item For all $2\leq n \leq N-1$ and $|\lambda|\neq 0$, the ground state projections satisfy 
		\be\label{MM3}
		\|g_{n+1}^{\infty}E_n^{\infty}\| \leq f(|\lambda|^2/2) := \sup_{n\geq 4}f_n(|\lambda|^2/2)
		\ee
		where given $\beta_n = \|\vp_{n-1}\|^2/\|\vp_n\|^2$, see \eqref{alpha_convergence},
		\be\label{def:fn}
		f_n(r) : = r\beta_n\beta_{n-2}\left(\frac{[1-\beta_{n-1}(1+r)]^2}{1+2r}+\frac{2(1-\beta_{n-1})^2}{1+r}\right).
		\ee
	\end{enumerate}
	As a consequence, if $|\lambda|>0$ and $f(|\lambda|^2/2)<1/3$ then the spectral gap of $H_\Lambda$ in $\caC_\Lambda^{\infty}$ is bounded from below by
	\be\label{tiling_gap}
	E_1(\caC_{[1,L]}^{\infty})  \geq \frac{2\kappa}{3}(1-\sqrt{3f(|\lambda|^2/2)})^2.
	\ee
	Moreover, $f(|\lambda|^2/2)<1/3$ for $|\lambda|\leq 7.4.$
\end{theorem}

\subsection{Restrictions of BVMD spaces to subvolumes}\label{sec:truncations}

For the proof of Theorem~\ref{thm:tiling_gap}, one needs to analyze the action on $\caC_\Lambda^{\infty}$ by operators $A$ supported on subintervals $\Lambda'\subseteq \Lambda$. It is therefore useful to expand $\caC_\Lambda^{\infty}$ as a direct sum of tensor products of the form $\cK_{\Lambda_l}\otimes\caC_{\Lambda'}(R)\otimes\cK_{\Lambda_r}$ where $\Lambda = \Lambda_l\cup \Lambda' \cup \Lambda_r$ is the disjoint union of three consecutive intervals, and $\cK_{\Lambda_\#}\subseteq \cH_{\Lambda_\#}.$ The main observation that allows us to write $\caC_\Lambda^{\infty}$ in such a form is that for any tiling $T\in\cT_\Lambda$  the restriction of the configuration $\sigma_\Lambda(T)$ to $ \Lambda' $ satisfies the requirements of Lemma~\ref{lem:tiling_configs}. Therefore,
\be\label{tiling_restriction}
\sigma_\Lambda(T)\restriction_{\Lambda'} = \sigma_{\Lambda'}(T') \quad \text{for some} \quad T'\in\cT_{\Lambda'}
\ee
and one sees that if $T$ is a $\bZ $-induced tiling then so too is $T'$ as every site holds at most two particles, see Figure~\ref{fig:restrictions}.
%\[(\sigma_{\Lambda}(T)\restriction_{\Lambda_l}, \sigma_{\Lambda'}(T''), \sigma_{\Lambda}(T)\restriction_{\Lambda_r}) = \sigma_{\Lambda}(\tilde{T}) \quad \text{for some} \quad \tilde{T}\leftrightarrow T.
%\]
%Here, we again use Lemma~\ref{lem:tiling_configs}. 
To state the desired decomposition of $\caC_\Lambda^{\infty}$, we introduce its orthonormal configuration basis
\begin{align}
	\cB_{\Lambda}^{\infty} & = \left\{\ket{\mu} \, | \, \mu\in\ran\sigma_\Lambda\cap\{0,1,2\}^\Lambda\right\} \label{NATile-basis}.
\end{align}

\begin{figure}
	\begin{center}
		\includegraphics[scale=.3]{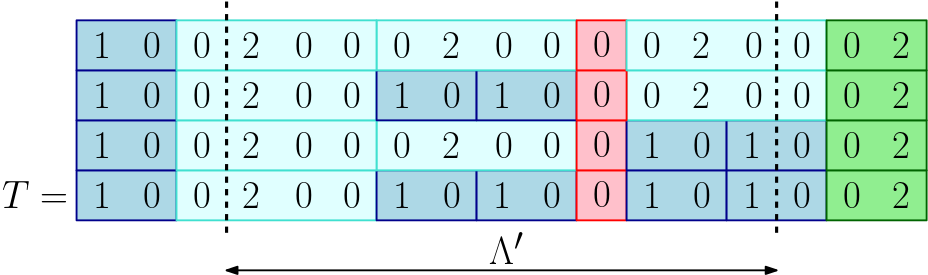}
	\end{center}
\caption{The restriction $T'$ of a tiling $T$ to $\Lambda'$. Inserting any connected tiling $T''\leftrightarrow T'$ produces a tiling on $\Lambda$ that is connected to $T$.}
\label{fig:restrictions}
\end{figure}

\begin{lem}[Tiling Space Decomposition]\label{lem:tiling_decomp}
	Let $\Lambda = [a,b]$, and suppose that $\Lambda'\subseteq \Lambda$ is a subinterval with $|\Lambda'|\geq 4$. Then, $C_\Lambda^{\infty}$ can be decomposed as
	\be\label{tiling_decomp}
	\caC_{\Lambda}^{\infty} = \bigoplus_{R'\in\cR_{\Lambda'}^{\infty}}\bigoplus_{\substack{\ket{\mu}\in\cB_{\Lambda}^{\infty} \, : \\ \mu = (\mu^l, \, \sigma_{\Lambda'}(R'), \, \mu^r)}} \ket{\mu^l}\otimes\caC_{\Lambda'}(R')\otimes\ket{\mu^r},
	\ee	
	where $\mu^l$ and $\mu^r$ are the subconfigurations of $\mu$ supported on the subinterval of $\Lambda$ to the left and right of $\Lambda'$, respectively. In the case that one of these subintervals is empty, we use the convention $\ket{\mu^\#} = 1$.
\end{lem}
Above, we use a slight abuse of notation and denote $\cS\otimes \psi:=\{\phi\otimes\psi:\phi\in\cS\}\subseteq \cH_1\otimes\cH_2$ for a subset $\cS\subseteq\cH_1$ and a vector $\psi\in\cH_2$ of two Hilbert spaces. Note that the direct sum in \eqref{tiling_decomp} is well-defined since it is taken over a collection of mutually orthogonal subspaces. Moreover, every root tiling $R'\in\cR_{\Lambda'}^\infty$ is represented in at least one summand on the RHS of \eqref{tiling_decomp} since any tiling configuration $\sigma_{\Lambda'}(T')$ can be extended by zeros to a tiling configuration on $\Lambda$.

\begin{proof}
	Fix $R'\in\cR_{\Lambda'}^{\infty}$ and pick any $\mu = (\mu^l,\sigma_{\Lambda'}(R'),\mu^r)\in \ran\sigma_{\Lambda}\cap\{0,1,2\}^\Lambda$. Applying the replacement rules to neighboring monomers in $R'$ to create $T'\in\cT_{\Lambda'}(R')$ once again produces a configuration $\mu(T')=(\mu^l,\sigma_{\Lambda'}(T'),\mu^r)$ that satisfies the conditions of Lemma~\ref{lem:tiling_configs}, see Figure~\ref{fig:restrictions}. Moreover, this configuration is $ \bZ $-induced since each site holds at most two particles. Hence, $\ket{\mu(T')}\in\cB_{\Lambda}^{\infty}$ for all $T'\in\cT_{\Lambda'}(R')$, and one finds that the RHS of \eqref{tiling_decomp} is a subspace of $\caC_\Lambda^{\infty}$.
	
	Now, fix any $\ket{\mu}\in\cB_{\Lambda}^{\infty}$, and decompose $\mu =(\mu^l,\mu^{\Lambda'},\mu^r)$ where $\mu^{\Lambda'}$ is the subconfiguration associated with $\Lambda'$. Since $\mu=\sigma_\Lambda(T)$ for some BVMD-tiling $T$ on $\Lambda$, $\mu^{\Lambda'}=\sigma_{\Lambda'}(T')$ for some $T'\in\cT_{\Lambda'}$  by \eqref{tiling_restriction}. Moreover, this tiling is $ \bZ $-induced as $\mu_x^{\Lambda'}\leq 2$ for all $x\in\Lambda'$. Thus,  
	\[
	\ket{\mu} \in \ket{\mu^l}\otimes\caC_{\Lambda'}(R')\otimes\ket{\mu^r}
	\]
	where $R'\in\cR_{\Lambda'}^{\infty}$ is the root-tiling associated to $T'$. Once again, $\mu(R'):=(\mu^l,\sigma_{\Lambda'}(R'),\mu^r)$ is a $ \bZ $-induced BVMD-tiling of $\Lambda$ since applying the replacement rules to any dimer ($D$ or $D^{(1)}$) in $T'$ reproduces a configuration on $\Lambda$ that satisfies Lemma~\ref{lem:tiling_configs}, see Figure~\ref{fig:restrictions}. Thus, $\caC_\Lambda^{\infty}$ is a subspace of the RHS of \eqref{tiling_decomp} and equality holds as claimed.
\end{proof}

The following is an immediate consequence of the above decomposition.
\begin{cor}[Subspace Reductions]\label{cor:reductions} Suppose $A$ is a self-adjoint operator supported on a subinterval $\Lambda'\subseteq\Lambda$ with $|\Lambda'|\geq 4$, and $\caC_{\Lambda'}(R')\subseteq\dom(A)$ is an invariant subspace of $A$ for each $R'\in\cR_{\Lambda'}^{\infty}$. Then $\caC_\Lambda^{\infty}\subseteq \dom(A\otimes\1_{\Lambda\setminus\Lambda'})$ is an invariant subspace of $A\otimes\1_{\Lambda\setminus\Lambda'}$, and the following properties hold:
	\begin{enumerate}
		\item $\|A\otimes \1_{\Lambda\setminus\Lambda'}\|_{\caC_\Lambda^{\infty}}=\|A\|_{\caC_{\Lambda'}^{\infty}}$, where the subscript denotes the Hilbert space in which the operator norm is taken.
		\item  $\spec(A\otimes \1_{\Lambda\setminus\Lambda'}\restriction_{\caC_\Lambda^{\infty}}) = \spec(A\restriction_{\caC_{\Lambda'}^{\infty}}).$
	\end{enumerate}
\end{cor}

Above, we use the notation $A\otimes\1_{\Lambda\setminus\Lambda'}$ to emphasize which Hilbert space we are considering the action of $A$. We suppress this notation in the proof below. Note also that the norms are well defined since $\caC_{\Lambda}^{\infty}$ is finite-dimensional for any finite volume $\Lambda$. 

\begin{proof} Since each $\caC_{\Lambda'}(R')$ is invariant under $A$, the latter is block diagonal with respect to the decomposition in \eqref{tiling_decomp} as
	\[
	A\left(\ket{\mu^l}\otimes\caC_{\Lambda'}(R)\otimes\ket{\mu^r}\right) = \ket{\mu^l}\otimes \left(A\caC_{\Lambda'}(R)\right)\otimes\ket{\mu^r}.
	\]
As a consequence, the norm and spectrum of the restrictions agree, i.e.
	\[
	\|A\|_{\ket{\mu_l}\otimes\caC_{\Lambda'}(R')\otimes\ket{\mu_r}} = \|A\|_{\caC_{\Lambda'}(R')} \quad\text{and}\quad \spec(A\restriction_{\ket{\mu_l}\otimes\caC_{\Lambda'}(R')\otimes\ket{\mu_r}}) = \spec(A\restriction_{\caC_{\Lambda'}(R')}),
	\]
 The claimed equalities then follow from applying the mutual orthogonality of the BVMD-spaces to conclude
	\be \label{norm_max}
	\|A\|_{\caC_{\Lambda'}^{\infty}} = \max_{R'\in\cR^\infty_{\Lambda'}}\|A\|_{\caC_{\Lambda'}(R')} \quad \mbox{and} \quad \spec(A\restriction_{\caC_{\Lambda'}^{\infty}}) = \bigcup_{R'\in\cR_{\Lambda'}^{\infty}}\spec(A\restriction_{\caC_{\Lambda'}(R')}),
	\ee	
	and similarly for $\|A\|_{\caC_\Lambda^\infty}$ and $\spec(A\restriction_{C_\Lambda^\infty})$ given \eqref{tiling_decomp}.
\end{proof}

A natural question to ask is how \eqref{tiling_decomp} relates to the trivial decomposition~\eqref{def:IVBVMD}.
As will be evident in the proof of Theorem~\ref{thm:tiling_gap}, the particular case of interest is when $\Lambda = [a,b]$ and $\Lambda'=[a,b-2]$. It is easy to see for this situation that every $\caC_{\Lambda'}(R')\otimes\ket{\mu^r}$ as in \eqref{tiling_decomp} is contained in some $\caC_\Lambda(R)$ with $R\in\cR_{\Lambda}^{\infty}$. More can be said, though, as illustrated in the next result. Specifically, we show that every $\caC_{\Lambda}(R)$ decomposes as a direct sum of one or two subspaces of the form $\caC_{\Lambda'}(R')\otimes\ket{\mu^r}$. This result is again derived from the possible ways tilings on $\Lambda$ can restrict to tilings on $\Lambda'$. 
There are two cases one needs to consider, which are distinguished by whether or not the replacement rules apply to the last two tiles in $R\in\cR_\Lambda^{\infty}$. We denote by
\be\label{R_MM}
\cR_\Lambda^{MM} = \{ R\in \cR_\Lambda^{\infty} \, | \, R \text{  ends in two or more monomers}\}
\ee
the set of $ \bZ$-induced root tilings for which the last two tiles can be replaced. For any tiling $R\in\cR_\Lambda^{MM}$ there is a unique choice $n\geq 2$ and $i\in\{1,2\}$ so that 
\be\label{two_monomers}
R = (\tilde{R},M_n^{(i)}) ,
\ee
where $\tilde{R}$ does not end in a monomer, and we recall that $M_n^{(i)}=(M,\ldots,M,M^{(i)})$ stands for a tiling of an interval of length $2(n-1)+i$ by $n$ monomers (the last of which has length $i\in\{1,2\}$).  We use the convention that $\tilde{R} = \emptyset$ if $R=M_n^{(i)}$. With respect to this decomposition, define the tiling $R_D \leftrightarrow R$ by replacing the last two monomers of $R$ with a dimer, 
\be\label{R_D}
R_D = (\tilde{R}, \, M_{n-2}^{(2)}, \, D^{(i)})
\ee
where $D^{(2)} = D$ and $R_D = (\tilde{R},D^{(i)})$ if $n=2$. Even though $R_D$ is not a root tiling of $\Lambda,$ its restriction produces a root tiling on $\Lambda'$ that is $\bZ$-induced.

\begin{lem}[BVMD-Space Decomposition]\label{lem:tiling_decomp2}
	Suppose $\Lambda = [a,b]$ and $\Lambda' = [a,b-2]$ with $|\Lambda'|\geq 4$. For any $R\in \cR_{\Lambda}^{\infty}\setminus \cR_\Lambda^{MM}$,
	\be\label{R_decomp_1}
	\caC_\Lambda(R) = \caC_{\Lambda'}(R')\otimes \ket{\mu}
	\ee
	where $\sigma_\Lambda(R) = (\sigma_{\Lambda'}(R'),\mu)$ and $R'\in\cR_{\Lambda'}^{\infty}$. Moreover, for any $R\in \cR_\Lambda^{MM}$,
	\be\label{R_decomp_2}
	\caC_\Lambda(R) = \big( \caC_{\Lambda'}(R')\otimes \ket{\mu_R}\big)  \oplus \big( \caC_{\Lambda'}(R_D')\otimes \ket{\mu_{R_D}}\big)
	\ee
	where $\sigma_\Lambda(R) = (\sigma_{\Lambda'}(R'),\mu_R)$, $\sigma_\Lambda(R_D) = (\sigma_{\Lambda'}(R_D'),\mu_{R_D})$ for $R_D$ as in \eqref{R_D}, and both $R',R_D'\in\cR_{\Lambda'}^{\infty}$.
\end{lem}
\begin{figure}
	\begin{center}
		\includegraphics[scale=.45]{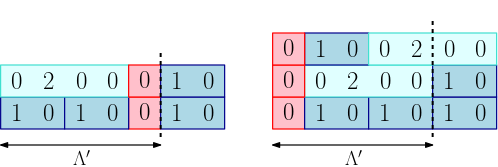}
	\end{center}
\caption{Examples of the restrictions to $\Lambda'$ for root tilings in $\cR_\Lambda^\infty\setminus \cR_{\Lambda}^{MM}$ and $\cR_\Lambda^{MM}$ respectively.}
\label{fig:restrictions2}
\end{figure}
This result will be proved by showing that the configuration bases agree. As such, denote by
\begin{eqnarray}
	\cB_{\Lambda}(R) = \left\{\ket{\sigma_\Lambda(T)} \, | \, T\in\cT_\Lambda(R)\right\}, \label{BVMD-basis}
\end{eqnarray}
the orthonormal basis of $\caC_\Lambda(R)$ with $ R \in \cR_\Lambda $.
\begin{proof}
Consider the two cases separately.
	
	\emph{Case $R\in\cR_\Lambda^{\infty}\setminus\cR_{\Lambda}^{MM}$:} In this case, the replacement rules used to generate the set of tilings $\cT_\Lambda(R)$ will never change the particle content of the last two sites of $\Lambda$, see Figure~\ref{fig:restrictions2}. As a consequence, every tile replacement on $\Lambda$ is in one-to-one correspondence with a tile replacement on $\Lambda'$. Thus,
	\[
	\cB_{\Lambda}(R) = \cB_{\Lambda'}(R')\otimes\ket{\mu}
	\]
	where $\mu = \sigma_\Lambda(R)\restriction_{[b-1,b]}$, and $R'$ is the root-tiling associated to $\sigma_{\Lambda}(R)\restriction_{\Lambda'}$.

	\emph{Case $R\in\cR_\Lambda^{MM}$:} Consider first the case that $R = (\tilde{R},\,M_n^{(2)})$ for some $n\geq 2$ and $\tilde{R}$ as in \eqref{two_monomers}. The particle content of the last two sites for any tiling $T\in\cT_\Lambda(R)$ is either $(1,0)$ if these two sites are covered by a monomer, or $(0,0)$ if the last two monomers are replaced by a bulk dimer. Considering all possible tilings on $\Lambda$, one quickly finds
	\be\label{k=2}
	\cB_\Lambda(R)= \big( \cB_{\Lambda'}(R')\otimes\ket{10} \big) \cup \big( \cB_{\Lambda'}(R_D')\otimes\ket{00}\big) 
	\ee
	where $R' = (\tilde{R},M_{n-1}^{(2)})$ and $R_D' = (\tilde{R},M_{n-2}^{(2)},B_2^r)$, see Figure~\ref{fig:restrictions2}. 
	
	The analogous argument holds when $R = (\tilde{R},\,M_n^{(1)})$, for which
	\be\label{k=1}
	\cB_\Lambda(R)= \big( \cB_{\Lambda'}(R')\otimes\ket{01}\big)  \cup\big( \cB_{\Lambda'}(R_D')\otimes\ket{20}\big) 
	\ee
	where $R' = (\tilde{R},M_{n-1}^{(1)})$ and $R_D' = (\tilde{R},M_{n-2}^{(2)},V)$. 
\end{proof}

A useful corollary for establishing \eqref{MM3} identifies a special orthogonal basis of $\caC_{\Lambda}^{\infty}\cap\big( \caG_{\Lambda'}\otimes \cH_{[b-1,b]}\big) $ with $\Lambda$ and $\Lambda'$ as in Lemma~\ref{lem:tiling_decomp2}. To state the result, we first recall that  $\caG_\Lambda\subseteq \caG_{\Lambda'}\otimes\cH_{[b-1,b]}$ by frustration-freeness, and so
\be\label{ff_containment}
\{\psi_\Lambda(R) \, | \, R\in\cR_\Lambda^{\infty}\} \subseteq \caC_{\Lambda}^{\infty}\cap \big( \caG_{\Lambda'}\otimes\cH_{[b-1,b]}\big) 
\ee
is an orthogonal set of vectors, see~\eqref{BVMD}. Using Lemma~\ref{lem:tiling_decomp2} we extend this set to an orthogonal basis in Corollary~\ref{cor:OB} by adding a set of vectors $\{\xi_\Lambda(R) \, | \, R\in \cR_\Lambda^{MM}\}$, which result from  decomposing $R = (\tilde{R},M_n^{(i)})\in\cR_\Lambda^{MM}$ as in \eqref{two_monomers}. Specifically,
\be\label{xi}
\xi_\Lambda(R) := \psi_{\Lambda(n,i)}(\tilde{R})\otimes \eta_n^{(i)} 
\ee
where $\psi_{\Lambda(n,i)}(\tilde{R})$ is the associated BVMD-state on $\Lambda(n,i):=[a,b-2(n-1)-i]$, 
\be\label{eta}
\eta_{n+1}^{(i)} := -\frac{\overline{\lambda}}{\sqrt{2}}\beta_{n} \vp_{n}\otimes\vp_1^{(i)} + \vp_{n-1}\otimes\ket{\sigma_d^{(i)}} \in\caC_{[1,2n+i]}(M_{n+1}^{(i)}) ,
\ee
and the ingredients defining the RHS above are as in Subsection~\ref{subsec:BVMD_properties}, see specifically \eqref{TT}, \eqref{recursion}, and \eqref{alpha_convergence}. This state is chosen so that $\braket{\eta_n^{(i)}}{\vp_n^{(i)}}=0$ for all $n\geq 2$ and $i\in\{1,2\}$. Like the squeezed Tao-Thouless states, $\eta_n^{(i)}$ is not normalized, but satisfies
\be\label{eta_norm}
\|\eta_n^{(k)}\|^2 = \|\vp_{n-2}\|^2\left[1+\beta_{n-1}\frac{|\lambda|^2}{2}\right]=\frac{\|\vp_{n-3}\|^2}{\beta_n\beta_{n-2}}.
\ee

\begin{cor}[Orthogonal Basis] \label{cor:OB}Let $\Lambda = [a,b]$ with $|\Lambda|\geq 7$, and $\Lambda' = [a,b-2]$. Then, the following is an orthogonal basis for $\caC_{\Lambda}^{\infty}\cap (\caG_{\Lambda'}\otimes\cH_{[b-1,b]})$:
	\be\label{orthogonal_basis}
	 \left\{\psi_{\Lambda}(R) \, | \, R\in\cR_\Lambda^{\infty}\right\}\cup \left\{\xi_{\Lambda}(R)\, | \ R\in\cR_\Lambda^{MM}\right\}.
	\ee
\end{cor}

\begin{proof} Just as in \eqref{gs_direct_sum}, the mutual orthogonality of the BVMD-spaces and the direct sum decomposition from \eqref{def:IVBVMD} guarantee that
	\[
	\caC_{\Lambda}^{\infty}\cap (\caG_{\Lambda'}\otimes\cH_{[b-1,b]}) = \bigoplus_{R\in\cR_{\Lambda}^{\infty}} \caC_\Lambda(R)\cap(\caG_{\Lambda'}\otimes\cH_{[b-1,b]}).
	\]
Since each $\caC_{\Lambda'}(R')$ supports a unique ground state of $\caG_{\Lambda'}$, by Lemma~\ref{lem:tiling_decomp2}  \[\dim\left(\caC_\Lambda(R)\cap\caG_{\Lambda'}\otimes\cH_{[b-1,b]}\right) = \begin{cases} 1 & \mbox{if} \; R\in \cR_\Lambda^{\infty}\setminus\cR_{\Lambda}^{MM}, \\ 2 & \mbox{if} \;  R\in \cR_{\Lambda}^{MM} .
\end{cases}
\]
Given \eqref{ff_containment}, one only needs to consider $R=(\tilde{R},M_n^{(i)})\in\cR_\Lambda^{MM}$ to complete the orthogonal basis. Using the notation from Lemma~\ref{lem:tiling_decomp2} (see also \eqref{k=2}-\eqref{k=1}) one can verify that for such $R$, 
\begin{align}
\psi_{\Lambda'}(R')\otimes\ket{\mu_R} & = \psi_{\Lambda(n,i)}(\tilde{R})\otimes\vp_{n-1}\otimes\vp_1^{(i)} \\
\psi_{\Lambda'}(R_D')\otimes\ket{\mu_{R_D}} & = \psi_{\Lambda(n,i)}(\tilde{R})\otimes\vp_{n-2}\otimes\ket{\sigma_d^{(i)}} 
\end{align}
and so $\xi_\Lambda(R)\in\caC_\Lambda(R)\cap(\caG_{\Lambda_1}\otimes\cH_{[b-1,b]}).$ By construction $\xi_\Lambda(R)$ and $\psi_\Lambda(R)$ are orthogonal since $\psi_\Lambda(R)=\psi_{\Lambda(n,i)}(\tilde{R})\otimes\vp_{n}^{(i)}$ and $\braket{\vp_n^{(i)}}{\eta_n^{(i)}} = 0$. Thus, the result holds as stated.
\end{proof}

We conclude this subsection with the following lemma, which constitutes the core of the proof of \eqref{MM3} in Theorem~\ref{thm:tiling_gap}.

\begin{lem}[Overlap]\label{lem:epsilon}
	Fix  $\Lambda = [a,b]$ with $|\Lambda|\geq 7$ and set $\Lambda_1 =[a,b-2]$ and $\Lambda_2 = [b-5,b]$. The ground-state projections satisfy the norm bound
	\be
	\label{norm_bound}
	\|G_{\Lambda_2}(\1-G_\Lambda)G_{\Lambda_1}\|_{\caC_{\Lambda}^{\infty}} \leq \sqrt{f(|\lambda|^2/2)}.
	\ee
	where $f$  is as in Theorem~\ref{thm:tiling_gap}.
\end{lem}

\begin{proof} The subspace $\cG_{\Lambda_1}^\infty := \caC_\Lambda^\infty\cap (\caG_{\Lambda_1}\otimes\cH_{[b-1,b]})$ is of the form considered in Corollary~\ref{cor:OB}. Comparing the orthogonal basis from that result with the orthogonal basis for $\cG_\Lambda$ in Theorem~\ref{thm:gss}, it is clear that $\caG_\Lambda^\perp \cap \cG_{\Lambda_1}^\infty = \spa \{\xi_\Lambda(R) |  R\in\cR_\Lambda^{MM}\}$ and
	\be\label{norm}
	\|G_{\Lambda_2}(\1-G_\Lambda)G_{\Lambda_1}\|_{\caC_{\Lambda}^{\infty}}^2 
	= 
	\sup_{0\neq \psi \in\caG_\Lambda^\perp \cap \cG_{\Lambda_1}^\infty }
	\frac{\|G_{\Lambda_2}\psi\|^2}{\|\psi\|^2}.
	\ee
	
	Recalling the identification from \eqref{idenfification}, $G_{\Lambda_2}\psi$ for any $\psi\in\cG_\Lambda^\perp\cap\caG_{\Lambda_1}^\infty$ can be expanded via Theorem~\ref{thm:gss} as
	\be\label{G_action}
	G_{\Lambda_2}\psi =  \sum_{R\in\cR_{\Lambda_2}^{\infty}}\frac{\ketbra{\psi_{\Lambda_2}(R)} }{\|\psi_{\Lambda_2}(R)\|^2}  \psi.
	\ee
	where we need only sum over $R\in\cR_{\Lambda_2}^\infty$, since $\psi$ is supported on $ \bZ $-induced tiling states. We first compute $G_{\Lambda_2}\xi_\Lambda(R)$ for an arbitrary $R\in\cR_\Lambda^{MM}$ and use this to bound $\|G_{\Lambda_2}\psi\|^2$ for an arbitrary state $\psi\in\caG_\Lambda^\perp \cap \cG_{\Lambda_1}^\infty$. Recalling the factored form $\xi_\Lambda(R) = \psi_{\Lambda(n,i)}(\tilde{R})\otimes \eta_n^{(i)}$ from \eqref{xi} and denoting by $\Gamma(n,i):=\Lambda\setminus \Lambda(n,i)$ the support of $\eta_n^{(i)}$, we consider two cases distinguished by $\Gamma(n,i)\subseteq \Lambda_2$, which holds for $n\leq 3$, and $\Lambda_2 \subsetneq \Gamma(n,i)$, which holds for $n\geq 4$.

Assume $n\leq 3$. Then $\Gamma(n,i)\subseteq \Lambda_2$, and so $G_{\Lambda_2}=G_{\Lambda_2}G_{\Gamma(n,i)}$ by frustration-freeness, and hence
\[
G_{\Lambda_2}\xi_\Lambda(R)= G_{\Lambda_2}\left(\psi_{\Lambda(n,i)}(\tilde{R})\otimes G_{\Gamma(n,i)}\eta_n^{(i)}\right) = 0 ,
\]
where the final equality holds since $\eta_{n}^{(i)}\in \caC_{\Gamma(n,i)}(M_n^{(i)})$ is orthogonal to the unique ground state $\vp_{n}^{(i)}\in \caC_{\Gamma(n,i)}(M_n^{(i)})$, and so the pairwise orthogonality of the BVMD-tiling spaces and analogous expansion from \eqref{G_action} for $\Gamma(n,i)$ guarantees
\[
G_{\Gamma(n,i)}\eta_n^{(i)} = \frac{\ketbra{\vp_n^{(i)}}\eta_n^{(i)} \rangle}{\|\vp_n^{(i)}\|^2}= 0.
\]

\begin{figure}
	\begin{center}
		\includegraphics[scale=.35]{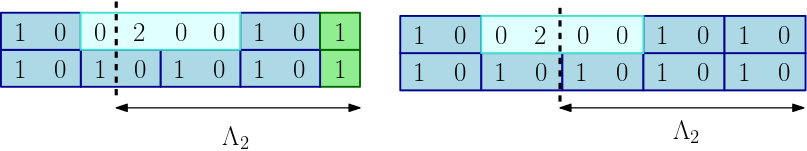}
	\end{center}
	\caption{The root tilings $R^i$ and $R^i_D$ for $i=1,2$, respectively.}
	\label{fig:Lambda_2_truncation}
\end{figure}	
If $n\geq4$, then $\Lambda_2 \subsetneq \Gamma(n,i)$ and we need only consider $G_{\Lambda_2}\eta_n^{(i)}$ since by \eqref{idenfification}
	\be\label{ngeq4}
	G_{\Lambda_2}\xi_\Lambda(R) = \psi_{\Lambda(n,i)}(\tilde{R})\otimes G_{\Lambda_2}\eta_n^{(i)}.
	\ee
To this end, notice that the restriction of any tiling $T\leftrightarrow M_n^{(i)}$ to $\Lambda_2$ produces a tiling $T'$ connected to one of two root tilings, $R^i,R_D^i\in\cR_{\Lambda_2}^\infty$, determined by whether or not $T$ has a dimer laying across the boundary of $\Lambda_2$ as in Figure~\ref{fig:Lambda_2_truncation}. Concretely, the root tilings are: 
\be\label{G_roots}
R^i = \begin{cases}
(V,M_3^{(1)}) & i = 1 , \\
M_3^{(2)} & i= 2 , 
\end{cases}
\qquad
R^i_D = \begin{cases}
(B_2^l,M_2^{(1)}), & i = 1 , \\
(V,V,M_2^{(2)}) & i= 2 .
\end{cases}
\ee

%To this end, we decompose $\caC_{\Gamma(n,i)}(M_n^{(i)})$ into a direct sum of BVMD-spaces on subintervals by truncating the tilings to $\Lambda_2$. Since truncating through a dimer creates $\bZ$-induced boundary tiles that can be used in root-tilings, the restriction of $M_n^{(i)}$ and $(M_{n-4}^{(2)},D,M_2^{(i)})$ to $\Lambda_2$ produces root tilings $R^{(i)},\,R_D^{(i)}\in\cR_{\Lambda_2}^{\infty}$ and $\tilde{R}^{(i)},\,\tilde{R}_D^{(i)}\in\cR_{\Gamma(n,i)\setminus\Lambda_2}^{\infty}$ for which the configuration basis of $\caC_{\Gamma(n,i)}(M_n^{(i)})$ decomposes into the disjoint union
%\be\label{basis_decomp}
%\cB_{\Gamma(n,i)}(M_n^{(i)}) = \cB_{\Gamma(n,i)\setminus\Lambda_2}(\tilde{R}^{(i)})\otimes\cB_{\Lambda_2}(R^{(i)}) \uplus \cB_{\Gamma(n,i)\setminus\Lambda_2}(\tilde{R}_D^{(i)})\otimes\cB_{\Lambda_2}(R_D^{(i)}) .
%\ee
%Concretely, the root tilings on $\Lambda_2$ are: $R^{(1)}=(V,M_3^{(1)})$, $R_D^{(1)}=(B_2^l,M_2^{(1)})$ for $i=1$, and $R^{(2)}=M_3^{(2)}$, $R_D^{(2)}=(V,V,M_2^{(2)})$ for $i=2$.
%%\begin{align}
%%R^{(1)} & = (V,M_3^{(1)}), \;\;\; R_D^{(1)} = (B_2^l,M_2^{(1)}) \\
%%R^{(2)} & = M_3^{(2)}, \qquad\,\;  R_D^{(2)} = (V,V,M_2^{(i)})
%%\end{align}
	
Using \eqref{G_action} to evaluate $G_{\Lambda_2}\eta_n^{(i)}$, the mutual orthogonality of the BVMD-spaces combined with \eqref{G_roots} reduce the calculation to
	\begin{align}\label{G_on_eta}
		G_{\Lambda_2}\eta_n^{(i)} & = \sum_{R'\in\{R^{i},R_D^{i}\}}\frac{\ketbra{\psi_{\Lambda_2}(R')}}{\|\psi_{\Lambda_2}(R')\|^2}\left(-\frac{\overline{\lambda}}{\sqrt{2}}\beta_{n-1} \vp_{n-1}\otimes\vp_1^{(i)} + \vp_{n-2}\otimes\ket{\sigma_d^{(i)}}\right)
	\end{align}
	where we have inserted the expansion of $\eta_n^{(i)}$ from \eqref{eta}. Applying the recursion relations~\eqref{recursion_general}-\eqref{recursion} to further expand $\eta_n^{(i)}$ and $\psi_{\Lambda_2}(R')$, one can compute \eqref{G_on_eta} to find
	\begin{align}\label{G2_calculation}
		G_{\Lambda_2}\eta_n^{(i)} = \frac{\overline{\lambda}\left(1-\beta_{n-1}\|\vp_2\|^2\right)}{\sqrt{2}\|\vp_3^{(i)}\|^2}\vp_{n-3}\otimes\vp_{3}^{(i)} + \frac{|\lambda|^2\left(1-\beta_{n-1}\right)}{2\|\vp_2^{(i)}\|^2}\vp_{n-4}\otimes\ket{\sigma_d}\otimes\vp_2^{(i)}
	\end{align}
	which is a sum of orthogonal vectors from $\caC_{\Gamma(n,i)}(M_n^{(i)}).$ The coefficients in \eqref{G2_calculation} are independent of $i$ by \eqref{vp2_to_vp1}. Calculating these explicitly and applying \eqref{eta_norm} then produces
	\[
	\|G_{\Lambda_2}\eta_{n}^{(i)}\|^2 = f_n(|\lambda|^2/2)\|\eta_n^{(i)}\|^2 \quad \mbox{with}\quad f_n(r) = r\beta_n\beta_{n-2}\left(\frac{[1-\beta_{n-1}(1+r)]^2}{1+2r}+\frac{r(1-\beta_{n-1})^2}{1+r}\right).
	\]
		
	 We further conclude that the action of $G_{\Lambda_2}$ on $\{\xi_\Lambda(R) | R\in\cR_\Lambda^{MM}\}$ preserves orthogonality since $G_{\Lambda_2}\xi_\Lambda(R) \in \caC_\Lambda(R)$ for any $R = (\tilde{R},M_n^{(i)})\in\cR_\Lambda^{MM}$ by \eqref{ngeq4} and \eqref{G2_calculation}. In addition, for each such $R$ the previous equality implies
	\[
	\|G_{\Lambda_2}\xi_\Lambda(R) \|^2 \leq \sup_{m\geq4}f_m(|\lambda|^2/2)\|\psi_{\Lambda(n,i)}(\tilde{R})\|^2\|\eta_n^{(i)} \|^2 = f(|\lambda|^2/2)\|\xi_\Lambda(R)\|^2 ,
	\]
	and so combining these two observations shows $\|G_{\Lambda_2}\psi\|^2 \leq f(|\lambda|^2/2) \|\psi\|^2$ for any $\psi\in\cG_\Lambda^\perp\cap\cG_{\Lambda_1}^\infty$, and the claim holds by \eqref{norm}.
\end{proof}

\subsection{Proof of Theorem~\ref{thm:tiling_gap}}\label{sec:proofMM}	
We now prove the main result of Section~\ref{sec:MM}.  First, if Conditions 1-3 of Theorem~\ref{thm:tiling_gap} hold, then applying \cite[Theorem~5.1]{nachtergaele:2016b} to the restriction $H_N^{\infty} = H_N\restriction_{\caC_\Lambda^{\infty}}$ produces,
\[
\inf_{0\neq \psi\in\caC_\Lambda^\infty \cap \ker(H_N)^\perp}\frac{\braket{\psi}{H_N^\infty\psi}}{\|\psi\|^2} \geq 2\kappa\left(1-\sqrt{3f(|\lambda|^2/2)}\right)^2,
\]
and \eqref{tiling_gap} follows by \eqref{equivalent_hams}. Thus, we only need to verify the three claims and the bound on $ f $.

1. For all $n$, $\caC_\Lambda^{\infty}$ is an invariant subspace of $h_n$. Its restriction $ h_n^{\infty} $ is defined as in~\eqref{blocks},
and the associated ground state projection is given by $g_n^{\infty}=  P_{\caC_{\Lambda}^{\infty}} g_n P_{\caC_{\Lambda}^{\infty}}$. Using Corollary~\ref{cor:reductions} it follows that
\[
h_n^{\infty} \geq E_1(\caC_{\Lambda_n}^{\infty})(\1-g_n^{\infty}) \geq \min_{k=5,6}E_1(\caC_{[1,k]}^{\infty})(\1-g_n^{\infty})
\]
where we have used translation invariance and $|\Lambda_n|\in\{5,6\}$ for the second inequality. Condition 1 then holds after computing the above minimum. Recalling that
\[
E_1(\caC_{[1,k]}^\infty) = \min\{E_1(\caC_{[1,k]}(R)) |  R\in\cR_{[1,k]}^\infty\},
\]
we compute $E_1(\caC_{[1,k]}(R))$ for all possible choices for $R$. By convention, $E_1(\caC_\Lambda(R)) =\infty$ for any root where $\dim(\caC_\Lambda(R)) = 1$ as such subspaces are contained in the ground-state space. The condition $\dim(\caC_\Lambda(R))>1$ requires that $R$ has two or more neighboring monomers.  Up to a factor of $\kappa$, the restriction $H_{[1,k]}\restriction_{\caC_{[1,k]}(R)}$ for any root with exactly two consecutively monomers is unitarily equivalent to the matrix from \eqref{MM_action}, which has a gap of $\kappa(|\lambda|^2+2)$. The restriction $H_{[1,k]}\restriction_{\caC_{[1,k]}(R)}$ for any root with three consecutive monomers is unitarily equivalent to the matrix
\be\label{eq:bmatrix}
\begin{bmatrix}
	2\kappa|\lambda|^2 & -\sqrt{2}\kappa\overline{\lambda} & -\sqrt{2}\kappa\overline{\lambda}\\
	-\sqrt{2}\kappa\lambda & 2\kappa & 0 \\
	-\sqrt{2}\kappa\lambda & 0 & 2\kappa
\end{bmatrix},
\ee
which has a gap $2\kappa$. This verifies Condition 1 since an interval $[1,k]$ with $k\leq 6$ can hold at most three monomers.

2. First, notice that the operators $g_n$ and $E_m=G_m-G_{m+1}$ as in Section~\ref{sec:MM} are defined so that
\[
\supp(g_n) = \Lambda_n, \quad \supp(E_m) \subseteq \supp(H_{m+1}) = [1,2m+2+k].
\]
Moreover, by construction $\ran(g_n) = \ker(H_{\Lambda_n})$ and $\ran(G_m) = \ker(H_{[1,2m+k]})$ and so by frustration-freeness $\ker(h_n)\subseteq \ker(H_m)$ for all $m\geq n$. As a consequence, $g_nG_m=G_m$ and $[g_n,E_m]=0$ for all $m\geq n$. In addition, $[g_n,E_m]=0$  for $m\leq n-4$ since $\supp(E_m)\subseteq[1,2m-6+k]$ which is disjoint from $\Lambda_n$. Summarizing,
\be\label{comm}
[g_n,E_m]\neq 0 \; \mbox{only if} \; m\in[n-3,n-1].
\ee
Since $\caC_\Lambda^{\infty}$ is an invariant subspace of the operators $g_n$ and $E_m$, they can be block diagonalized as in \eqref{blocks}, and one concludes that the commutator relations in \eqref{comm} are inherited by the respective blocks $E_m^{\infty}$ and $g_n^{\infty}$. Hence, Condition 2 holds.

3. Fix $2\leq n\leq N-1$. Since $\caC_\Lambda^{\infty}$ is an invariant subspace of both $g_{n+1}$ and $E_{n}$, \eqref{blocks} applies to both operators and $g_{n+1}^\infty E_n^\infty = (g_{n+1}E_n)\restriction_{\caC_\Lambda^\infty}$. The product $g_{n+1}E_n$ is supported on $[1,2n+2+k]$, and so Corollary~\ref{cor:reductions} yields
\[
\|g_{n+1}^{\infty}E_n^{\infty}\| = \|g_{n+1}E_n\|_{\caC_\Lambda^{\infty}} = \|g_{n+1}E_n\|_{\caC_{[1,2n+2+k]}^{\infty}}.
\]
Once again applying frustration-freeness, the product factored as
\[
g_{n+1}E_n = G_{\Lambda_{n+1}}(G_{[1,2n+k]}-G_{[1,2n+2+k]}) = G_{\Lambda_{n+1}}(\1-G_{[1,2n+2+k]})G_{[1,2n+k]},
\]
which is of the form considered in Lemma~\ref{lem:epsilon}. This produces
\[
\|g_{n+1} E_n\|_{\caC_{[1,2n+2+k]}^{\infty}} \leq \sqrt{f(|\lambda|^2/2)},
\]
verifying Condition 3. 

4. The function $f(r)$ was analyzed in \cite[Appendix A]{NWY:2020}, where it was shown that $f(r^2)<1/3$ for $r\in[0,5.3]$. The claimed bound on $f$ is immediate from taking $r=|\lambda|/\sqrt{2}$.

\section{Proof of the uniform bulk gap}\label{sec:UBG}

The goal of this section is to present a method that establishes the robust bulk gap for the periodic Hamiltonian $ H_\Lambda^\textrm{per} $ asserted in Theorem~\ref{thm:main2}. In order to avoid the issue of edge states, we will proceed as sketched in Subsection~\ref{sec:invariant}. Using
\begin{equation}\label{eq:bulkgapest}
	E_1^\textrm{per}(\cH_\Lambda) = \min\left\{ E_1^\textrm{per}(\caC_\Lambda^\textrm{per}) , E_0^\textrm{per}\left(\big(\caC_\Lambda^\textrm{per}\big)^\perp\right)\right\} , 
\end{equation}
we reduce the proof of the lower bound on the spectral gap to proving separately a lower bound on the spectral gap of $  H_\Lambda^\textrm{per} $  restricted to the invariant subspace $ \caC_\Lambda^\textrm{per} $ of periodic tiling states introduced in Subsection~\ref{sec:periodic_tilings}, and a lower bound on the ground-state energy in the complement subspace. We start with a bound on the former quantity, and follow up with electrostatic estimates on the latter in Subsection~\ref{sec:electro2}. 

The a priori decomposition of the Hilbert space in terms of the invariant subspace $ \caC_\Lambda^\textrm{per}$ such that \eqref{eq:bulkgapest} holds is key to the proof of the uniform bulk gap in Theorem~\ref{thm:main2} since the Hamiltonian $ H_{\Lambda'} $ with open boundary conditions on any interval $ \Lambda' \subsetneq  \Lambda $ is free of edge states in this subspace. As long as one can find such a decomposition, this idea is robust and also applicable to the bulk gap question for other models with edge states, such as the one studied in~\cite{NWY:2021}.

\subsection{Finite-volume criteria avoiding edge states}

As in the application of the martingale method in Section~\ref{sec:MM}, for the finite-size criterion we will consider the restriction of operators $A:\cH_\Lambda\to\cH_\Lambda$ supported on subintervals $\Lambda'\subsetneq\Lambda$ to the periodic tiling space $\caC_\Lambda^\per$. In the case that $\caC_\Lambda^\per$ is an invariant subspace of $A$, the operator can be block diagonalized similar to \eqref{blocks} as
\be\label{per_blocks}
A := A\restriction_{\caC_\Lambda^\per} + (\1-P_\Lambda^\per)A(\1-P_\Lambda^\per), \quad A\restriction_{\caC_\Lambda^\per} = P_\Lambda^\per A P_\Lambda^\per 
\ee
where $P_\Lambda^\per$ is the orthogonal projection onto $\caC_\Lambda^\per$. In particular, such a decomposition holds for the Hamiltonian $H_{\Lambda'}$ as well as its ground state projection $G_{\Lambda'}$ for any subinterval $\Lambda' \subseteq \Lambda$ in the ring geometry.

An analysis similar to that of Section~\ref{sec:truncations} also applies to restrictions of the periodic tiling space $\caC_\Lambda^\per$ to subintervals $\Lambda'\subseteq\Lambda$. 
The truncation of any periodic tiling on the ring $\Lambda$ to an interval $\Lambda'\subseteq\Lambda$ produces a $\bZ$-induced tiling on $\Lambda'$. 
 Conversely, if $|\Lambda|\geq |\Lambda'|+3$ every $\bZ$-induced tiling on $\Lambda'$ can be realized as a truncation of a periodic tiling, from which the following variation of Lemma~\ref{lem:tiling_decomp} can be deduced. Namely,
 \[ 
 \caC_\Lambda^\textrm{per} = \bigoplus_{R'\in\cR_{\Lambda'}^{\infty}}\bigoplus_{\substack{\ket{\mu}\in\cB_{\Lambda}^\textrm{per} \, : \\ \mu = ( \sigma_{\Lambda'}(R'), \, \mu^{\Lambda\backslash \Lambda'} )}} \caC_{\Lambda'}(R')\otimes\ket{ \mu^{\Lambda\backslash \Lambda'} },
 \]
 where $ \cB_\Lambda^\textrm{per} $ is the orthonormal configuration basis of periodic tilings. 
 The size constraint here guarantees that bulk tiles can be placed in $\Lambda$ to produce all possible combinations of $\bZ$-induced boundary tiles on $\Lambda'$. A minor adaptation of the argument in Corollary~\ref{cor:reductions} then produces the following isospectral relation for any subinterval $ \Lambda' \subset \Lambda $ of the ring geometry such that $|\Lambda'|\leq |\Lambda|-3$:
\begin{equation}\label{eq:isospectral}
	\spec \left(H_{\Lambda'} \otimes \mathbbm{1}_{\Lambda\backslash \Lambda'} \restriction_{  \caC_\Lambda^\textrm{per}  }\right) \ = \ \spec \left(H_{\Lambda'}  \restriction_{  \caC_{\Lambda'}^\infty }\right)  \ = \ \spec \left(H_{\Lambda'}^\infty\right) . 
\end{equation}
%A further isospectral argument showing equality of the spectrum restricted to various BVMD-spaces can be used to extend the above equality to all subintervals such that $|\Lambda'|<|\Lambda|$.

Of particular interest, the spectral gap of the restriction on the LHS of \eqref{eq:isospectral} agrees with the spectral gap of the restriction on the RHS:
\begin{equation}\label{eq:iso2}
	\inf_{0\neq\psi\in \caC_\Lambda^\textrm{per}\cap (\caG_{\Lambda'}^\perp \otimes \cH_{\Lambda\backslash \Lambda'} )  }  \frac{\braket{\psi}{H_{\Lambda'}\psi}}{\|\psi\|^2}   = E_1(\caC_{\Lambda'}^\infty) .
\end{equation}
A lower bound on the RHS was the topic of Theorem~\ref{thm:tiling_gap}. Moreover, we recall from the proof of this theorem in Subsection~\ref{sec:proofMM} that the spectral gap of $ H_{\Lambda'}^\infty  $ on any interval $ \Lambda' $ of size five or six is at least $ 2 \kappa $, and the operator norm of this restriction is bounded by the largest eigenvalue of the matrix~\eqref{eq:bmatrix}, i.e.
\begin{equation}\label{def:gamma}
	\gamma :=  \inf_{|\Lambda'| \in \{ 5,6 \} }  E_1(\caC_{\Lambda'}^\infty) = 2\kappa , \qquad   \Gamma:= \sup_{|\Lambda'| \in \{ 5,6 \} } \left\| H_{\Lambda'}^\infty \right\| = 2 \kappa (1 + |\lambda|^2) .
\end{equation}

With this in mind, we may now formulate the following finite-volume criterion for the gap $ E_1^\textrm{per}(\caC_\Lambda^\textrm{per}) $. This criterion is an improvement of~\cite[Theorem 3.11]{NWY:2021}. Its proof below closely follows the strategy used in this previous work and relies on the version of Knabe's method~\cite{knabe:1988} from~\cite[Theorem~3.10]{NWY:2021}. 
\begin{theorem}[Periodic Spectral Gap]\label{thm:pbc_gap} Fix $n\geq 2$. Then for any ring $\Lambda$ such that $|\Lambda| \geq 3n+6$,
	\begin{equation}\label{eq:pergap3}
		E_1^\per(\caC_{\Lambda}^\per)  \geq \frac{\gamma\,  n}{2\Gamma(n-1)}\left[\min_{1\leq l\leq3}  E_1(\caC_{[ 1,3n  + l]}^\infty)-\frac{\Gamma}{n}\right] . 
	\end{equation}
\end{theorem}
\begin{proof}
	By translation invariance, it is sufficient to consider the ring associated to $\Lambda := [1,3N+r]$ where $N>0$ and $r\in\{2,3,4\}$ are the unique integers so that $|\Lambda|=3N+r$. Let $\Lambda_1, \ldots, \Lambda_{N+1}$ be the sequence of subintervals of $\Lambda$ defined by
	\begin{align*}
	\Lambda_{i}  & = [3i-2,3i+2] \;\; \text{for} \;\; 1\leq i \leq N \;\; \text{and} \;\; \Lambda_{N+1} = \begin{cases}
		\left[|\Lambda|-2,|\Lambda|+2\right], & r=2,3 \\
		\left[|\Lambda|-3,|\Lambda|+2\right], & r=4
	\end{cases}
	\end{align*}
where we identify $x\equiv x+|\Lambda|$. These intervals are defined so that every interaction term, $n_xn_{x+1}$ or $q_x^*q_x$, that contributes to $H_\Lambda^\per$ is supported on at least one and at most two of the intervals $\Lambda_i$. As a consequence, for each $1\leq k\leq N+1$ 
\begin{equation}
	H_{\Lambda_{n,k}}\leq \sum_{i=k}^{n+k-1} H_{\Lambda_i} \leq 2H_{\Lambda_{n,k}} , \quad
	H_{\Lambda}^{\per} \leq \sum_{i=1}^{N+1}H_{\Lambda_i} \leq 2H_{\Lambda}^\per \label{eq:first_bound}, 
\end{equation}
where the addition $n+k-1$ is taken modulo $N+1$ and $\Lambda_{n,k}:=\bigcup_{i=k}^{n+k-1}\Lambda_i$. As $\caC_\Lambda^\per$ is an invariant subspace of each of the Hamiltonians in \eqref{eq:first_bound}, the same relations hold when one restricts all of the Hamiltonians to this subspace.

Let $P_i:\caC_\Lambda^\per \to \caC_\Lambda^\per$ denote the orthogonal projection onto $\ran(H_{\Lambda_i}\restriction_{\caC_\Lambda^\per})$. Thus, applying the isospectrality from~\eqref{eq:isospectral} 
\begin{equation}\label{eq:projection_bounds}
	\gamma P_i \leq H_{\Lambda_i} \restriction_{\caC_\Lambda^\textrm{per}} \leq \Gamma P_i, 
\end{equation}
with $ \gamma $ and $ \Gamma $ as in~\eqref{def:gamma} since each interval $\Lambda_i$ contains 5 or 6 sites. Summing \eqref{eq:projection_bounds} over appropriate values of $i$ and using the restricted versions of \eqref{eq:first_bound} produces the operator inequalities
\begin{equation}\label{sf_inequalites}
	\frac{\gamma}{2} H_{n,k}  \leq  \, H_{\Lambda_{n,k}}\restriction_{\caC_\Lambda^\textrm{per}} \, \leq \Gamma H_{n,k} , \quad \frac{\gamma}{2} H_N  \, \leq H_{\Lambda}^{\per}\restriction_{\caC_\Lambda^\textrm{per}}\, \leq \Gamma  H_N 
\end{equation}
where the operators $H_{n,k}$ and $H_N$ on $\caC_\Lambda^\per$ are defined by
\[ H_{n,k} := \sum_{i=k}^{n+k-1}P_i, \qquad  H_N := \sum_{i=1}^{N+1} P_i .
\]

Depending on the value of $r$ and whether or not $\Lambda_{N+1}$ contributes to the definition of $\Lambda_{n,k}$, one can easily deduce that $\Lambda_{n,k}$ is an interval with at least $3n+1$ sites and at most $3n+3$ sites. Thus, the gap of $H_{\Lambda_{n,k}}\restriction_{\caC_\Lambda^\textrm{per}} $ satisfies~\eqref{eq:iso2} since $|\Lambda|\geq |\Lambda_{n,k}| +3 $. Combining this with \eqref{sf_inequalites}, it readily follows by translation invariance that for all $ k \in \{1, \ldots, N+1\} $:
\begin{equation}\label{eq:gap_inequalities}
	\Gamma E_{1}^{(n,k)} \geq E_1(\caC_{\Lambda_{n,k}}^\infty) \geq  \min_{1\leq l \leq 3}  E_1(\caC_{[ 1, 3n  + l]}^\infty) \quad \text{and}\quad E_1^\textrm{per}(\caC_{\Lambda}^\textrm{per})  \geq \frac{\gamma}{2} E_{1}^L 
\end{equation}
where $E_{1}^{(n,k)}$ and $E_1^L$ denote the spectral gaps of $H_{n,k}$ and $H_L$ in $\caC_\Lambda^\per$, respectively. The proof is completed after applying \cite[Theorem~3.10]{NWY:2021} to produce a lower bound on $E_1^L$ using $E_1^{(n,k)}$. For its application, we note that 
the Hamiltonian $H_L$ is frustration-free, since \eqref{sf_inequalites} implies $\ker(H_L)=\ker(H_\Lambda^\per)$, which is nontrivial by Theorem~\ref{thm:periodic_gss}. In addition, $P_i\psi = (\1-G_{\Lambda_i})\psi$ for all $\psi\in\caC_{\Lambda}^\per$ and $i=1,\ldots, N+1$ by \eqref{per_blocks}. Since operators defined on $\cH_\Lambda$ that are supported on disjoint spatial regions commute, and the intervals $\Lambda_i$ were chosen so that $\Lambda_{i}\cap \Lambda_{j}= \emptyset$ unless $|i-j|=1$ or $\{i,j\}=\{1,N+1\}$, this guarantees that $[P_i,P_j] =  0$ under the same constraints on $i$ and $j$. Therefore, the operators $P_i$ satisfy the requirements of~\cite[Theorem~3.10]{NWY:2021} on the Hilbert space $ \caC_{\Lambda}^\textrm{per} $, and hence the respective spectral gaps $E_1^L$ and $E_1^{(n,k)}$ satisfy the bound
	\begin{equation}\label{eq:proj_bound}
		E_{1}^L \geq \frac{n-1}{n}\left(\min_{1\leq k \leq N+1} E_{1}^{(n,k)} -\frac{1}{n}\right) .
	\end{equation}
Using \eqref{eq:gap_inequalities} to further bound \eqref{eq:proj_bound} produces the result.
\end{proof}

\subsection{Electrostatic estimates for periodic boundary conditions}\label{sec:electro2}
For our proof of the bulk gap via~\eqref{eq:bulkgapest} we also need a lower bound on the ground-state energy $ E_0^\textrm{per}((\caC_\Lambda^\per)^\perp)$ of the Hamiltonian $ H_\Lambda^\textrm{per} $ restricted to the invariant subspace orthogonal to periodic tiling states. The approach we employ to bound this energy is inspired by the related question of a lower bound on the \emph{Yrast line} and, specifically, how the ground state energy $E_0(N)$ for the periodic system on $\Lambda$ in the $N$-particle sector behaves as the number of particles increases. In the physical regime $|\lambda|^2\ll 1$ and for  total filling $\nu:=N/|\Lambda|>1$, a positive lower bound on $E_0(N)$ that scales with $\nu$ can easily be deduced from the following Cauchy-Schwarz operator inequality
\be\label{CS}
q_x^*q_x \geq (1-\delta)n_x(n_x-1) - |\lambda|^2\frac{1-\delta}{\delta}n_{x-1}n_{x+1} 
\ee 
which holds for all $ \delta \in (0,1) $.
\begin{proposition}[Yrast Line]\label{prop:Yrast} Fix an interval $\Lambda$ and suppose that $|\lambda|^2<(\kappa-2)/(2\kappa)$. Then for any $\nu>1$,
\be\label{eq:Yrast}
E_0(\nu|\Lambda|)\geq  \nu|\Lambda|\left[\nu(1+\kappa/2-\kappa|\lambda|^2)-\kappa/2\right].
\ee
\end{proposition}
\begin{proof}
Applying \eqref{CS} with $\delta = 1/2$ for all $x$ produces an operator lower bound on $H_\Lambda^\per$ in terms of a sum of electrostatic operators. The claimed bound on $E(\nu|\Lambda|)$ results from minimizing the energy over $\mu\in[0,\infty)^{\Lambda}$ of the classical problem
\[
\cE(\mu):=\sum_{x\in\Lambda}\mu_x\mu_{x+1} + \frac{\kappa}{2} \mu_x(\mu_x-1) -\kappa|\lambda|^2\mu_x\mu_{x+2} \quad \text{subject to}\quad \sum_{x\in\Lambda}\mu_x = \nu|\Lambda|.
\]
When $|\lambda|^2<(\kappa-2)/(2\kappa)$ this has a unique minimum at $\mu_x = \nu$ for all $x\in\Lambda$, which produces \eqref{eq:Yrast}.
\end{proof}

While the simple operator inequality \eqref{CS} is sufficient for estimating the ground state energy of $H_\Lambda^\per$ in sufficiently high particle sectors, to bound $E_0^\per((\caC_\Lambda^\per)^\perp)$ one also needs to consider configurations $\mu\in\bN_0^\Lambda$ with low filling. Our approach here is to use a refined version of \eqref{CS} that depends on the type of configuration under consideration.  

To this end, recall that any state $\psi\in \big( \caC_\Lambda^\textrm{per} \big)^\perp $ is a linear combination of $ \mu \in  \bN_0^\Lambda\setminus\big( \ran \sigma_\Lambda \restriction_{\cT_\Lambda^\textrm{per}}\big)  =: \mathcal{S}_\Lambda $, i.e.\
\[
\psi = \sum_{\mu \in  \mathcal{S}_\Lambda  } \psi(\mu)\ket{\mu}, \quad\text{where}\quad \|\psi\|^2 = \sum_{\mu \in  \mathcal{S}_\Lambda  }|\psi(\mu)|^2 <\infty.
\]
Given $ \Lambda = [a,b] $ with $|\Lambda|\geq 4$, Lemma~\ref{lem:periodic_configs} characterizes $\mu=(\mu_a, \ldots, \mu_b)\in  \mathcal{S}_\Lambda $ as those particle configurations which belong to one of the following disjoint sets (which are to be understood using the convention $ x \equiv x +|\Lambda| $):
\begin{align*} &  \cS_{\Lambda}^{(1)} := \left\{ \mu \;  \big| \;  \mu_x\mu_{x+1}\geq 1\;  \mbox{for some $x\in\Lambda$} \right\} , \\
	&   \cS_{\Lambda}^{(2)}  := \left\{ \mu  \;  \big| \;  \mu_x \geq 3 \;  \mbox{for some $x\in\Lambda$}\right\} \backslash   \cS_{\Lambda}^{(1)}  , \\
	&   \cS_{\Lambda}^{(3)}  := \left\{ \mu  \;  \big| \; \mu_x=\mu_{x+3} = 2 \;  \mbox{for some  $x\in\Lambda$}\right\} \backslash  ( \cS_{\Lambda}^{(1)}  \cup  \cS_{\Lambda}^{(2)}  )  , \\
	&  \cS_{\Lambda}^{(4)}  := \left\{\mu \  \; \big| \; \mu_x\mu_{x+2}\geq 2 \; \mbox{for some $x\in\Lambda$} \right\} \backslash  ( \cS_{\Lambda}^{(1)}  \cup  \cS_{\Lambda}^{(2)}   \cup  \cS_{\Lambda}^{(3)}  )  .
\end{align*}
Note that all possible violations of the conditions from Lemma~\ref{lem:periodic_configs}  are covered since (i) $\cS_\Lambda^{(1)}$ is the set of all configurations that violate Condition~1, (ii) $\cS_\Lambda^{(2)}$ is all configurations that satisfy Condition~1, but violate Condition~2 by filling a site with three or more particles, and (iii) $\cS_\Lambda^{(3)}\uplus \cS_\Lambda^{(4)}$ contains all configurations that satisfy Conditions~1, have at most two particles on every site, but still violate Condition~2. Thus, we have constructed a disjoint partition 
\be\label{non-tilings2}
 \mathcal{S}_\Lambda  = \cS_{\Lambda}^{(1)}   \uplus \cS_{\Lambda}^{(2)} \uplus \cS_{\Lambda}^{(3)}   \uplus \cS_{\Lambda}^{(4)}  . 
\ee

Any configuration $ \mu \in  \cS_{\Lambda}^{(1)}  $ has positive electrostatic energy
$
e_\Lambda^\textrm{per}(\mu) = \sum_{x=a}^{b} \mu_x \mu_{x+1} $, 
and the mean energy of any $\psi\in \big( \caC_\Lambda^\textrm{per}\big)^\perp\cap\dom(H_\Lambda^\textrm{per})$ is given by
\begin{equation}\label{eq:Cenergy2}
\langle \psi | H_\Lambda \psi \rangle = \sum_{\mu \in  \cS_{\Lambda}^{(1)}  } e_\Lambda^\textrm{per}(\mu)  |\psi(\mu)|^2 + \kappa \, \sum_{\nu \in \mathbb{N}_0^\Lambda} \sum_{x=a}^{b} \left| (q_x\psi)(\nu) \right|^2 .  
\end{equation}
Our strategy for a lower bound is to bound every term in $ \sum_{\mu \in   \mathcal{S}_\Lambda }|\psi(\mu)|^2 $ by a few terms from the RHS of~\eqref{eq:Cenergy2}. This is trivial for $ \mu \in  \cS_{\Lambda}^{(1)}  $ since those configurations have electrostatic energy. For all other $ \mu \in \mathcal{S}_\Lambda$ we will associate (i) a configuration $ \eta(\mu)  \in  \cS_{\Lambda}^{(1)} $ with electrostatic energy and (ii) one or two terms $|(q_x\psi)(\nu)|^2$ from the second part of~\eqref{eq:Cenergy} which will be bounded using a variation of \eqref{CS}. This yields the following result.
%Roughly speaking, the role of these additional terms is to relate $\psi(\eta(\mu))$ with $\psi(\mu)$ in such a way that the collective energy from all of the associated terms can be bounded from below by $\gamma_\kappa^\textrm{per}(|\lambda|^2)|\psi(\mu)|^2$. 

\begin{theorem}[Electrostatic Estimates I] \label{thm:electro2}
For any $\Lambda = [a,b] $ with $|\Lambda|\geq 8$, the ground state energy of $H_\Lambda^\textrm{per} $ in the invariant subspace $\big( \caC_\Lambda^\textrm{per} \big)^\perp $ satisfies the lower bound
	\begin{equation}\label{perp_gap}
		E_0^\textrm{per}(\caC_\Lambda^\perp) \geq \frac{1}{4} \min\left\{1,  \frac{2\kappa}{\kappa+1} ,    \frac{2 \kappa}{1+\kappa|\lambda|^2}    \right\} =\gamma_\kappa^\textrm{per}(|\lambda|^2)  .
	\end{equation}
\end{theorem} 
\begin{proof}
We consider $\psi\in \big( \caC_\Lambda^\textrm{per} \big)^\perp $ and fix any configuration $ \mu \in  \mathcal{S}_\Lambda $ in its support.

\textit{Case $ \mu \in \cS_{\Lambda}^{(1)}  $:} For such configurations we have the trivial bound
\begin{equation}\label{gamma1}
e_\Lambda^\textrm{per}(\mu) \, |\psi(\mu)|^2 \geq |\psi(\mu)|^2  =: \gamma^{(1)}|\psi(\mu)|^2. 
\end{equation}

\textit{Case $ \mu \in \cS_{\Lambda}^{(2)}  $:} 
Set $ x \equiv x_\mu := \max\{x\in[a,b]   \, | \, \mu_x \geq 3 \}  $, and note that this means $ \mu_{x-1} =  \mu_{x+1} = 0 $ since otherwise $ \mu  \in \cS_{\Lambda}^{(1)}  $. Recalling the notation from the beginning of Section~\ref{sec:main} (see the text following \eqref{creation-action}), we associate to $\mu$ the configurations $ \nu(\mu) := \alpha_x^2 \mu $ and $ \eta(\mu) := \alpha_{x-1}^*\alpha_{x+1}^* \nu(\mu) $, for which $ e_\Lambda^\textrm{per}(\eta(\mu)) \geq 2 $. Using \eqref{creation-action}, we then have for any $ \delta \in (0,1) $
\begin{align}
	T^{(2)} [\psi;\mu] :=  \left| (q_x\psi)(\nu(\mu)) \right|^2 & = \left|  \sqrt{\mu_x(\mu_x-1)} \psi(\mu) - \lambda \psi(\eta(\mu)) \right|^2 \notag \\
	 	& \geq (1-\delta) \mu_x(\mu_x-1) | \psi(\mu)|^2 - \frac{1-\delta}{\delta} |\lambda|^2 | \psi(\eta(\mu)) |^2 . \label{eq:case2}
\end{align}
Choosing $ \delta = \frac{\kappa|\lambda|^2}{2+\kappa|\lambda|^2}  $ and bounding $  \mu_x(\mu_x-1)  \geq 6 $ yields
\begin{equation}\label{gamma2}
	 e_\Lambda^\textrm{per}(\eta(\mu) )  | \psi(\eta(\mu)) |^2  + \kappa  T^{(2)} [\psi;\mu] \geq \frac{12 \kappa}{2+\kappa|\lambda|^2}   | \psi(\mu)|^2 =: \gamma^{(2)}|\psi(\mu)|^2 . 
\end{equation}

\textit{Case $ \mu \in \cS_{\Lambda}^{(3)}  $:} Let $ x \equiv x_\mu := \max\{x\in[a,b]   \, | \, \mu_x=\mu_{x+3}=2\}  $. Since $\mu\notin \cS_\Lambda^{(1)}$, it follows that $\mu_{x-1}=\mu_{x+1}=\mu_{x+2}=0$. To $ \mu $ we associate four configurations $ \nu(\mu) := \alpha_x^2 \mu $ and  $ \eta'(\mu) :=  \alpha_{x-1}^*\alpha_{x+1}^* \nu(\mu) $ as well as  $ \nu'(\mu):=  \alpha_{x+1}\alpha_{x+3} \eta'(\mu) $ and  $ \eta(\mu) := (\alpha_{x+2}^*)^2 \nu'(\mu) $. The last configuration has electrostatic energy $ e_\Lambda^\textrm{per}(\eta(\mu)) \geq 2$. Estimating similarly to \eqref{eq:case2}, we obtain for any $ \delta, \delta' \in (0,1) $,
 \begin{align}
	 T^{(3)} [\psi;\mu] & :=  \left| (q_x\psi)(\nu(\mu)) \right|^2 +  \left| (q_{x+2}\psi)(\nu'(\mu)) \right|^2 \label{eq:Case3}\\
	 & = \left|  \sqrt{2} \psi(\mu) - \lambda \psi(\eta'(\mu)) \right|^2 +  \left|  \sqrt{2} \psi(\eta(\mu))  - \lambda\sqrt{2} \psi(\eta'(\mu)) \right|^2 \nonumber\\
	 	 & \geq 2(1-\delta) | \psi(\mu)|^2 + |\lambda|^2 \left[ 2 (1-\delta') - \frac{1-\delta}{\delta} \right] |\psi(\eta'(\mu)) |^2 - 2\frac{1-\delta'}{\delta'}| \psi(\eta(\mu)) |^2.\nonumber
\end{align}
The choice $ \delta = \frac{1}{1+2(1-\delta')} $ eliminates the second term. In turn, choosing $ \delta' = \frac{\kappa}{\kappa+1} $ yields
\begin{equation}\label{S3_bound}
 e_\Lambda^\textrm{per}(\eta(\mu) )  | \psi(\eta(\mu)) |^2  + \kappa  T^{(3)} [\psi;\mu] \geq \frac{2\kappa}{\kappa+1}|\psi(\mu)|^2 =: \gamma^{(3)}|\psi(\mu)|^2.
\end{equation}

\textit{Case $ \mu \in \cS_{\Lambda}^{(4)}  $:} Choose $ x \equiv x_\mu := \max\{x\in[a,b]   \, | \, \mu_x=2  \; \text{and} \; \max\{ \mu_{x-2},\mu_{x+2} \} \geq 1  \}  $ and note that $ \mu_{x\pm 1} = 0 $.  We then proceed similarly as in the second case and associate to  $ \mu $ the configurations $ \nu(\mu):= \alpha_x^2\mu $ and $\eta(\mu) :=  \alpha_{x-1}^*\alpha_{x+1}^* \nu(\mu) $. The latter configuration has electrostatic energy $ e_\Lambda^\textrm{per}(\eta(\mu) ) \geq 1 $. Proceeding as in~\eqref{eq:case2} we then bound for any $ \delta \in (0,1) $:
\[
	T^{(4)} [\psi;\mu] :=  \left| (q_{x}\psi)(\nu(\mu)) \right|^2 \geq  2 (1-\delta)  | \psi(\mu)|^2 - \frac{1-\delta}{\delta} |\lambda|^2 | \psi(\eta(\mu)) |^2 .
\]
The choice $ \delta = \frac{\kappa|\lambda|^2}{1+\kappa|\lambda|^2}  $  then yields
\begin{equation}\label{eq:Case4}
e_\Lambda^\textrm{per}(\eta(\mu) )  | \psi(\eta(\mu)) |^2  + \kappa  T^{(4)} [\psi;\mu]  \geq \frac{2\kappa}{1+ \kappa |\lambda|^2}     | \psi(\mu)|^2 =: \gamma^{(4)}|\psi(\mu)|^2.  
\end{equation}
\bigskip

Employing the conventions $  T^{(1)} [\psi;\mu]  \equiv 0 $ and $ \eta(\mu) = \mu $ for $ \mu \in \cS_{\Lambda}^{(1)} $, the collective estimates above show that summing over all configurations in $\cS_\Lambda^{(j)}$ for each fixed $ j \in \{ 1,2,3,4 \} $ satisfies:
\begin{align}
	\gamma^{(j)} \sum_{\mu\in  \cS_{\Lambda}^{(j)} } |\psi(\mu)|^2 &  \leq  \sum_{\mu\in  \cS_{\Lambda}^{(j)} } e_\Lambda^\textrm{per}(\eta(\mu))  |\psi(\eta(\mu))|^2+ \kappa \sum_{\mu\in  \cS_{\Lambda}^{(j)} }   T^{(j)} [\psi;\mu] \nonumber  \\
	& \leq  c_j\sum_{\eta \in  \cS_{\Lambda}^{(1)} } e_\Lambda^\textrm{per}(\eta)  |\psi(\eta)|^2 + \kappa \, \sum_{\nu \in \mathbb{N}_0^\Lambda} \sum_{x=a}^{b} \left| (q_x\psi)(\nu) \right|^2  \leq c_j\langle \psi | H_\Lambda \psi \rangle , \label{j-estimate2}
\end{align}
where $c_j\in\bN$ counts the maximum number of times a configuration with electrostatic energy $\eta\in\cS_\Lambda^{(1)}$ is associated to a configuration $\mu\in\cS_\Lambda^{(j)}$, 
\[c_j := \max_{\eta\in\cS_{\Lambda}^{(1)}}\left|\left\{\mu\in\cS_\Lambda^{(j)} \big| \eta(\mu)=\eta\right\}\right|.\]
For \eqref{j-estimate2}, we also use that the set of pairs $(x,\nu)\in\Lambda\times\bN_0^\Lambda$ that contribute to $T^{(j)} [\psi;\mu]$ is in one-to-one correspondence with $\mu\in\cS_\Lambda^{(j)}$. The final bound in \eqref{perp_gap} is obtained by dividing by $c_j$, summing over $j$ and seeing that $4\gamma_\kappa^\textrm{per}(|\lambda|^2)=\min_j\gamma^{(j)}/c_j$. Thus, the proof is complete after determining the values of $c_j$.

It is trivial that $c_1=1$. For all other $j$ and fixed $\mu\in\cS_\Lambda^{(j)}$, the sites in $\Lambda$ that contribute to the electrostatic energy of $\eta(\mu)\in\cS_\Lambda^{(1)}$ are localized to an interval of size at most 6 around $x_\mu$. The constraint $|\Lambda|\geq 8$ guarantees that this interval can be uniquely identified in the ring geometry. By considering separately for each $j$ the possible forms of $\mu$ and $\eta(\mu)$ in this interval, one can quickly deduce that the mapping $S_\Lambda^{(j)}\ni \mu \mapsto \eta(\mu)\in \cS_\Lambda^{(1)}$ is injective  for $j=3,4$ giving $c_j = 1$ for those $j$, and that the preimage of any $\eta\in\cS_\Lambda^{(1)}$ for $j=2$ has at most two elements producing $c_2=2$.\footnote{The configurations $\mu\in\cS_\Lambda^{(2)}$ that map to a non-unique $\eta(\mu)$ are ones for which $\mu_x = 3$, $\{\mu_{x+2},\mu_{x-2}\}=\{0,1\},$ and $\mu_{x\pm3}=0$. All other configurations $\mu\in\cS_\Lambda^{(2)}$ are in a one-to-one correspondence with $\eta(\mu)\in\cS_\Lambda^{(1)}$.} 
\end{proof}

\subsection{Proof of Theorem~\ref{thm:main2}}\label{Sec:ProofMain2}

We are now ready to finalize the proof of Theorem~\ref{thm:main2} by bounding the RHS of~\eqref{eq:bulkgapest}. The lower bound on the ground-state energy $ E_0^\per((\caC_\Lambda^\per)^\perp) $ in Theorem~\ref{thm:electro2} yields the first term in the minimum on the RHS of~\eqref{eq:main2}. Using the finite-volume criterion from Theorem~\ref{thm:pbc_gap}, the spectral gap $E_1^\textrm{per} (\caC_\Lambda^\textrm{per})   $ in~\eqref{eq:bulkgapest}  is bounded from below in terms of the spectral gap $E_1(\caC_{[1,m]}^\infty)$ on subintervals of $ \Lambda $ of lengths $m\geq 7$. In turn, these gaps are uniformly lower bounded in Theorem~\ref{thm:tiling_gap}, which yields for all $ m \geq 7 $, 
\[
E_1(\caC_{[1,m]}^\infty) \geq  \frac{2\kappa}{3}(1-\sqrt{3f(|\lambda|^2/2)})^2.
\]
In the situation that $ f(|\lambda|^2/2)   < 1/3 $, the condition
\[
\inf_{l=1,2,3}E_1(\caC_{[1,3n+l]}^\infty) - \frac{\Gamma}{n} >0
\]
is thus satisfied for sufficiently large  $n $. Taking the limit $n\to\infty$ in the bound~\eqref{eq:pergap3}, and inserting the values~\eqref{def:gamma} produces the second term in the minimum on the RHS of~\eqref{eq:main2}. \qed

\section{Edge and excited states}

The aim of this section is twofold. The first goal is to complete  the proof of Theorem~\ref{thm:main} by producing a lower bound on the ground-state energy in the subspace orthogonal to all BVMD tilings. This is accomplished in the next section, where an analogue of the electrostatic estimate in Theorem~\ref{thm:electro2} is proved for open boundary conditions. This  bound is limited by the presence of  edge states, which are discussed and classified at the end of  Section~\ref{sec:Electrostatic}. 
As a second goal, we prove variational bounds on low-lying bulk excitations in Section~\ref{sec:Excited_States}, where a brief discussion on many-body scars of mid- and high-energy is also provided. 
%We now turn our attention to discussing low-lying excitations of the truncated Haldane pseudopotential. In particular, edge states will be discussed in Section~\ref{sec:Electrostatic} and the proof of Theorem~\ref{thm:main} will be given in Section~\ref{sec:eeproofmain}, and so we return to the setting of open boundary conditions. However, the analysis of bulk excitations in Section~\ref{sec:Excited_States} can be easily modified to also hold for periodic boundary conditions. A brief discussion on many body scars of mid- and high-energy is also given. 
%

\subsection{Electrostatic estimates and edge states for open boundary conditions}\label{sec:Electrostatic}
For our proof of the spectral gap of $ H_\Lambda \equiv H_\Lambda^\textrm{obc} $ via~\eqref{eq:twoway}, we still need to establish a volume-independent lower bound on the ground-state energy in the invariant subspace $  \caC_\Lambda^\perp$ orthogonal to all BVMD tiling states. Together with the bound on the spectral gap of $ E_1(\caC_{\Lambda}^\infty) $ from Section~\ref{sec:MM}, this then completes the proof of the uniform gap for open boundary conditions.

\begin{theorem}[Electrostatic estimate II]\label{thm:perp_bound} For any interval $\Lambda $ with $|\Lambda|\geq 5$, the ground state energy of $H_\Lambda$ in the invariant subspace $\caC_\Lambda^\perp$ satisfies the lower bound
	\begin{equation}\label{perp_gap_OBC}
	E_0(\caC_\Lambda^\perp) \geq \frac{1}{5} \min\left\{1,   \frac{2 \kappa}{1+\kappa|\lambda|^2} ,  \frac{2\kappa}{\kappa+1}  ,  \frac{2\kappa |\lambda|^2}{\kappa+1}  \right\} =\gamma_\kappa(|\lambda|^2)  .
	\end{equation}
\end{theorem}

The proof of this theorem parallels the one of Theorem~\ref{thm:electro2} for the periodic case. 
We will therefore focus on the subtle differences and provide the proof using the same notation. Similar to the previous case, any state $\psi\in \caC_\Lambda^\perp$ is a linear combination of configurations $ \mu \in  \bN_0^\Lambda\setminus\ran \sigma_\Lambda =: \mathcal{S}_\Lambda $, that is
\be\label{eq:vec_expansion}
\psi = \sum_{\mu \in  \mathcal{S}_\Lambda } \psi(\mu)\ket{\mu}, \quad\text{where}\quad \|\psi\|^2 = \sum_{\mu \in  \mathcal{S}_\Lambda }|\psi(\mu)|^2 <\infty.
\ee
Fixing $ \Lambda = [a,b] $ with $|\Lambda|\geq 5$, Lemma~\ref{lem:tiling_configs} characterizes $\mu=(\mu_a, \ldots, \mu_b)\in  \mathcal{S}_\Lambda $ as those particle configurations which belong to one of the following disjoint sets:
\begin{align*} &  \cS_{\Lambda}^{(1)} := \left\{ \mu  \;  \big| \;  \mu_x\mu_{x+1}\geq 1\;  \mbox{for some $x\in[a,b-1]$} \right\} , \\
&   \cS_{\Lambda}^{(2)}  := \left\{ \mu  \;  \big| \;  \mu_x \geq 3 \;  \mbox{for some $x\in[a+1,b-1]$}\right\} \backslash   \cS_{\Lambda}^{(1)}  , \\
&   \cS_{\Lambda}^{(3)}  := \left\{ \mu \;  \big| \; \min\{\mu_x, \,\mu_{x+3}\} = 2 \;  \mbox{for some  $x\in[a,b-3]$}\right\} \backslash  ( \cS_{\Lambda}^{(1)}  \cup  \cS_{\Lambda}^{(2)}  )  , \\
&  \cS_{\Lambda}^{(4)}  := \left\{\mu  \; \big| \; \mu_x= 2 \; \mbox{and} \; \max\{\mu_{x-2},\mu_{x+2}\}\geq 1 \; \mbox{for some $x\in[a+1,b-1]$} \right\} \backslash  ( \cS_{\Lambda}^{(1)}  \cup  \cS_{\Lambda}^{(2)}   \cup  \cS_{\Lambda}^{(3)}  )  \\
&  \cS_{\Lambda}^{(5)}  := \left\{\mu  \; \big| \; \mu_a\mu_{a+2}\geq 2 \; \mbox{or} \; \mu_{b-2}\mu_b\geq 2 \right\} \backslash  ( \cS_{\Lambda}^{(1)}  \cup  \cS_{\Lambda}^{(2)}   \cup  \cS_{\Lambda}^{(3)}\cup  \cS_{\Lambda}^{(4)}  )  .
\end{align*}
We use the convention $\mu_{a-1}=\mu_{b+1}=0$ in the definition of $\cS_\Lambda^{(4)}$. The last two sets are chosen as a partition of
\be\label{edge_configs}
\cS_{\Lambda}^{(4)} \cup \cS_{\Lambda}^{(5)} = \left\{\mu  \; \big| \; \mu_x\mu_{x+2}\geq 2 \; \mbox{for some $x\in[a,b-2]$} \right\} \backslash  ( \cS_{\Lambda}^{(1)}  \cup  \cS_{\Lambda}^{(2)}   \cup  \cS_{\Lambda}^{(3)}  ),
\ee
which is analogous to the fourth set of configurations from the periodic case in Section~\ref{sec:electro2}. Here, $\cS_\Lambda^{(4)}$ consists of the configurations from \eqref{edge_configs} where there is an \emph{interior} site $x$ with two particles that has an occupied next nearest neighbor. The set $\cS_\Lambda^{(5)}$ corresponds to the configurations from \eqref{edge_configs} where the site with two (or more) particles must be on the boundary and the next-nearest neighbor holds exactly one particle; it is precisely these configurations that produce the low-lying energies for small $|\lambda|$ which were discussed in Section~\ref{sec:main}.

Given~\eqref{edge_configs}, the construction of the sets $\cS_\Lambda^{(j)}$ only differ from their counterparts in the periodic case at the boundary of $ \Lambda $. Since all possible violations of the conditions in Lemma~\ref{lem:tiling_configs} are covered, we have constructed a disjoint partition of 
$ 
\mathcal{S}_\Lambda $. 
For open boundary conditions, a configuration $ \mu \in  \cS_{\Lambda}^{(1)}  $ has electrostatic energy
$
e_\Lambda(\mu) = \sum_{x=a}^{b-1} \mu_x \mu_{x+1} $, 
and the mean energy of any $\psi\in \caC_\Lambda^\perp\cap\dom(H_\Lambda)$ is given by
\begin{equation}\label{eq:Cenergy}
\langle \psi | H_\Lambda \psi \rangle = \sum_{\mu \in  \cS_{\Lambda}^{(1)}  } e_\Lambda(\mu)  |\psi(\mu)|^2 + \kappa \, \sum_{\nu \in \mathbb{N}_0^\Lambda} \sum_{x=a+1}^{b-1} \left| (q_x\psi)(\nu) \right|^2 .  
\end{equation}

\begin{proof}[Proof of Theorem~\ref{thm:perp_bound}] 
	We expand $\psi\in \caC_\Lambda^\perp\cap\dom(H_\Lambda)$ as in \eqref{eq:vec_expansion} and fix a configuration $ \mu \in \mathcal{S}_\Lambda  $. The argument for $\mu\in\cS_\Lambda^{(j)}$ with $j=1,2$ or $4$ proceeds identically as in the proof in Theorem~\ref{thm:electro2}: we define $\eta(\mu)$ and $T^{(j)}[\psi;\mu]$ as before with the only modification being that the value of $x_\mu$ is constrained to the interval of sites used to define the set $\cS_\Lambda^{(j)}$ above. In these cases, one again produces the bound
	\be\label{eq:same_bound}
	e_\Lambda(\eta(\mu) )  | \psi(\eta(\mu)) |^2  + \kappa  T^{(j)} [\psi;\mu] \geq \gamma^{(j)}|\psi(\mu)|^2  
	\ee
	with $\gamma^{(j)}$ defined as in the proof of Theorem~\ref{thm:electro2}, see specifically \eqref{gamma1}, \eqref{gamma2} and \eqref{eq:Case4}.
	
	\textit{Case $ \mu \in \cS_{\Lambda}^{(3)}  $:} The argument follows the analogous case from Theorem~\ref{thm:electro2} with only technical modifications. We set $ x \equiv x_\mu := \max\{x\in[a,b-3]   \, | \, \min\{\mu_x, \mu_{x+3}\}\geq2\}  $, and consider the cases $x>a$ and $x=a$ separately.
	
	For $x>a$, defining $\eta(\mu)$ and $T^{(3)}[\psi;\mu]$ exactly as in \eqref{eq:Case3}, and bounding similarly with the choices $ \delta' = \frac{\kappa}{\kappa+1} $ and $ \delta =(1+\mu_{x+3}(1-\delta'))^{-1} \leq (1+2 (1-\delta'))^{-1} $ once again produces \eqref{eq:same_bound} with $\gamma^{(3)}$ as in \eqref{S3_bound}. For open boundary conditions, it is possible to have $\mu_{x+3}>2$ if $x=b-3$, which accounts for the slightly different choice of $\delta$.
	
	The case $x=a$ runs analogously to that of $x>a$, with the roles of $x$ and $x+3$ interchanged since $x+3$ is in the interior of $\Lambda$. Setting $\nu(\mu) = \alpha_{x+3}^2\mu$, $\eta'(\mu)=\alpha_{x+2}^*\alpha_{x+4}^*\nu(\mu)$, $\nu'(\mu)=\alpha_x\alpha_{x+2}\eta'(\mu)$ and $\eta(\mu) = (\alpha_{x+1}^*)^2\nu'(\mu)$ we then estimate
	\[ 
	T^{(3)} [\psi;\mu] :=  \left| (q_{x+3}\psi)(\nu(\mu)) \right|^2 +  \left| (q_{x+1}\psi)(\nu'(\mu)) \right|^2 \]
	as in the case $x>a$ which again yields \eqref{eq:same_bound}.
	
	\textit{Case $ \mu \in \cS_{\Lambda}^{(5)}  $:} Due to the presence of a boundary, our strategy differs from the case of $\cS_\Lambda^{(4)}$, and we pick $ x \equiv x_\mu := \max\{x\in\{a+1,b-1\}   \, | \, \mu_{x-1}\mu_{x+1} \geq 2   \}  $ and note that $ \mu_{x} = 0 $.  To  $ \mu $ we associate the configurations $ \nu(\mu) := \alpha_{x-1} \alpha_{x+1} \mu $, and $ \eta(\mu) := (\alpha_x^*)^2  \nu(\mu) $. The latter configuration has electrostatic energy $ e_\Lambda(\eta(\mu) ) \geq 2 $. We then bound for any $ \delta \in (0,1) $:
	\begin{align*}
	T^{(5)} [\psi;\mu] :=  \left| (q_{x}\psi)(\nu(\mu)) \right|^2 & = \left|  \sqrt{2}  \psi(\eta(\mu)) - \lambda \sqrt{\mu_{x-1} \mu_{x+1}}  \psi(\mu)  \right|^2 \\
	& \geq (1-\delta)|\lambda|^2  \mu_{x-1}\mu_{x+1} | \psi(\mu)|^2 - 2\frac{1-\delta}{\delta} | \psi(\eta(\mu)) |^2 .
	\end{align*}
	The choice $ \delta = \frac{\kappa}{\kappa+1} $ then yields the final estimate
	\begin{equation}\label{eq:Case5_OBC}
	e_\Lambda(\eta(\mu) )  | \psi(\eta(\mu)) |^2  + \kappa  T^{(5)} [\psi;\mu]  \geq \frac{2\kappa |\lambda|^2}{\kappa+1}     | \psi(\mu)|^2 := \gamma^{(5)}|\psi(\mu)|^2.  
	\end{equation}
	
	Summing the collective estimates above over all configurations in $\cS_\Lambda^{(j)}$ for each fixed $ j \in \{ 1,2,3,4,5 \} $ once again produces
	\begin{align}
	\gamma^{(j)} \sum_{\mu\in  \cS_{\Lambda}^{(j)} } |\psi(\mu)|^2 &  \leq  \sum_{\mu\in  \cS_{\Lambda}^{(j)} } e_\Lambda(\eta(\mu))  |\psi(\eta(\mu))|^2+ \kappa \sum_{\mu\in  \cS_{\Lambda}^{(j)} }   T^{(j)} [\psi;\mu] 
	%		\nonumber  \\
	%		& \leq  c_j\sum_{\eta \in  \cS_{\Lambda}^{(1)} } e_\Lambda(\eta)  |\psi(\eta)|^2 + \kappa \, \sum_{\nu \in \mathbb{N}_0^\Lambda} \sum_{x=a+1}^{b-1} \left| (q_x\psi)(\nu) \right|^2  
	\leq c_j\langle \psi | H_\Lambda \psi \rangle , \label{j-estimate}
	\end{align}
	where $c_j\in\bN$ counts the maximum number of times a configuration with electrostatic energy $\eta\in\cS_\Lambda^{(1)}$ is associated to a configuration $\mu\in\cS_\Lambda^{(j)}$. Our choices are again made so that $c_j=1$ for $j\neq2$, and $c_2=2$. The proof then concludes by dividing \eqref{j-estimate} by $c_j$, summing over $j$, and noting that $5\gamma_\kappa(|\lambda|^2)=\min_j\gamma^{(j)}/c_j$. 
\end{proof}

While bound in Theorem~\ref{thm:perp_bound} is sufficient for our purposes, it is not optimal as far as constants are concerned. However,  the lower bound does scale as $ \mathcal{O}(|\lambda|^2) $ in the regime of small $ |\lambda| $, which agrees with the scaling of the edge states singled out in~\eqref{ex:edge_energy}.  Since this scaling is only reflected in the estimate from~\eqref{eq:Case5_OBC}, any such low-energy state $\psi$ must have a nonzero overlap with some configuration $\mu\in\cS_\Lambda^{(5)}$, i.e. $\psi(\mu)\neq 0$. We end this section by using modified tilings to identify the invariant subspaces that contain the configurations from $\cS_\Lambda^{(5)}$. These tilings will only differ from BVMD tilings at the boundary, and so in this sense any eigenstate with energy $\mathcal{O}(|\lambda|^2)$ can be interpreted as an edge state.

 Since every configuration $\mu\in\cS_\Lambda^{(5)}$ only breaks the conditions from Lemma~\ref{lem:tiling_configs} at the boundary of $\Lambda$, it is the configuration associated with a \emph{edge tiling} $T=(T_1,\ldots, T_k)$ of $\Lambda=[a,b]$ where either $T_1 = (n010)$ or $T_k=(10n)$ for some $n\geq 2$, and all other tiles are BVMD-tiles. In addition to the original BVMD tiles and replacement rules $(10)(10)\leftrightarrow(02000)$ and $(10)(1)\leftrightarrow(020)$, the action of the dipole hopping terms $q_x^*q_x$ on these edge tilings generates the following new set of boundary tiles and replacement rules:
\begin{enumerate}
	\item \emph{On the left boundary:} The edge tiles $(n010)$, $((n-1)200)$, and $(n0200)$ for $n\geq 2$ which satisfy the replacement rules
	\be\label{left_replacements}
	(n010)\leftrightarrow((n-1)200), \quad (n010)(10)\leftrightarrow(n00200)
	\ee
	\item \emph{On the right boundary:} The edge tiles $(10n)$, $(02(n-1))$, and $(0200n)$ for $n\geq 2$ which satisfy the replacement rules
	\be\label{right_replacements}
	(10n)\leftrightarrow(02(n-1)), \quad (10)(10n)\leftrightarrow(0200n)
	\ee
\end{enumerate}

With these modified edge tiles, the invariant subspaces that contain states with energy $\mathcal{O}(|\lambda|^2)$ are of the form $\cE_\Lambda(R)=\spa\{\sigma_\Lambda(T) \, | \,  T \leftrightarrow R\}$ where $R=(R_1,\ldots,R_k)$ is any tiling of $\Lambda$ so that
\[
R_1\in\{B_n^l, M, V, (n010) | n\geq 2\}, \;\; R_k \in\{B_n^r, M^{(1)}, M, V, (10n) |  n \geq 2\}, \;\; R_i \in\{M,V\} \;\; 1<i<k , 
\]
and at least one $R_1 = (n010)$ or $R_k=(10n)$. Here, the equivalence relation $T\leftrightarrow R$ is defined using the new replacement rules from \eqref{left_replacements}-\eqref{right_replacements} as well as the original BVMD-replacement rules. As remarked above, the restriction of any of these tilings to the interior $[a+1,b-1]$ produces a BVMD-tiling.

\subsection{Proof of Theorem~\ref{thm:main}}\label{sec:eeproofmain}
We  now spell out the short proof of Theorem~\ref{thm:main}. Due to \eqref{eq:twoway} one only needs to estimate the spectral gap in the tiling space $ E_1(\caC_\Lambda) $ as well as the ground state in the orthogonal complement $E_0(\caC_\Lambda^\perp)$. The latter has been accomplished in Theorem~\ref{thm:perp_bound}. The former is bounded from below using Theorem~\ref{thm:gap_reduction} together with the uniform estimate in Theorem~\ref{thm:tiling_gap}. \qed

\subsection{Variational estimates for excited states}\label{sec:Excited_States}

%As shown in Theorem~\ref{thm:perp_bound}, for small $|\lambda|$ the eigenstates with low, positive energy are supported on configurations that different from BVMD-tilings at the boundary. As such, we conjecture that the bulk gap of the Hamiltonian is not accurately reflected by the estimate produced in Theorem~\ref{thm:perp_bound}. 
The bound~\eqref{eq:main2} on the bulk gap avoids edge states, but is nevertheless not sharp even in the regime of small $ |\lambda| $.  
To investigate low-lying exited states in greater detail, we study the Hamiltonian $H_\Lambda$ restricted to other invariant subspaces spanned by tiling states $ | \sigma_\Lambda(T) \rangle $. 
We choose to work with with open boundary conditions. However, the analysis can easily be modified to periodic boundary conditions.

While the set of tilings $T$ is once again generated from a root tiling and a few substitution rules to keep the complexity manageable, the set of tiles is augmented to basic (bulk) tiles plus a few additional tiles. One simple class of states results from adding the new bulk tile $(01)$ to produce, e.g., roots of the form
\be\label{quasiparticle}
 R_{l,r}^{(m)} = (10)_l (01)_m (10)_r ,
\ee
where we place $ l $, respectively, $ r $, basic bulk monomers $(10)$ to the left, respectively, right, of a string of $ m $ tiles $(01)$ on the interval $ \Lambda  = \Lambda_l \cup \{0\} \cup \Lambda_{r}^{(m)}$ with $ \Lambda_l =  [-2l , -1] $ and $ \Lambda_r^{(m)} = [ 1,  2(r+m)-1] $. By allowing the boundary tile $M^{(1)}$ to be placed in the interior of $\Lambda$, \eqref{quasiparticle} can be represented in terms of the BVMD-tiles as $R_{l,r}^{(m)}=(10)_l(0)(10)_{m-1}(1)(10)_r$. The action of the operators $ q_x^* q_x $ with  $ -2l   < x < 2 (r+m) -1$ on this root then generate additional tiles and replacement rules analogous to~\eqref{eq:replacement} that are used to define tilings $T$ which are connected to the root, denoted $R_{l,r}^{(m)}\leftrightarrow T$, see \eqref{excited_rr1}-\eqref{excited_rr2} below. Due to the presence of a void at the interface $x=0$, the invariant subspace corresponding to this root factors, 
\begin{align*}
	& \mathcal{D}_{l,r}^{(m)}  :=  \spa\left\{ | \sigma_\Lambda(T) \rangle \, \big| \,  T \leftrightarrow R_{l,r}^{(m)}\right\} = \mathcal{C}_{\Lambda_l}(M_l^{(2)}) \otimes | 0 \rangle \otimes  \mathcal{D}_{r}^{(m)} \\
	& \mbox{with}\quad  \mathcal{D}_{r}^{(m)}:= \spa\left\{ | \sigma_{\Lambda_r^{(m)}}(T) \rangle \, \big| \, T \leftrightarrow R_{r}^{(m)}  = (M_{m-1}^{(2)}, M^{(1)}, M_r^{(2)})  \right\} ,
\end{align*} 
where $M_n^{(2)}$ covers an interval of length $2n$ with $n$ regular monomers $M$.
For all finite $ m \in \mathbb{N} $ the set of additional tiles generated from $ R_{r}^{(m)} $ by the replacement rules is finite. More specifically:
\begin{enumerate}
	\item	 For  $ m = 1 $, in addition to the usual replacement rule $ (10)(10) \leftrightarrow (0200) $, define the rule
	\be \label{excited_rr1}
	(1)(0200)  \leftrightarrow (02100) .
	\ee
	\item For $ m \geq 2 $, in addition to the usual replacement rules $ (10)(10) \leftrightarrow (0200) $ and $ (10)(1) \leftrightarrow (020)$, define
	\be \label{excited_rr2}
	(020) (10) \leftrightarrow (01200) \qquad (1)(0200)  \leftrightarrow (02100) .
	\ee
\end{enumerate}
\begin{figure}
	\begin{center}
		\includegraphics[scale=.25]{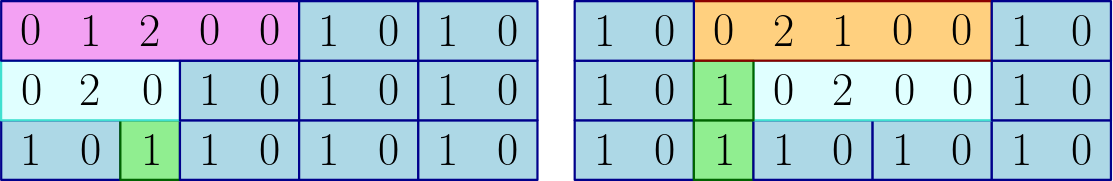}
	\end{center}
\caption{Some of the tilings generated by the replacement rules on $R_3^{(2)}$. Since the hopping terms act on three consecutive sites, the distortions caused by $M^{(1)}$ are localized to an interval of size seven.}
\end{figure}

The same strategy has been applied in~\cite{NWY:2020} to explore excited states of the truncated $ \nu=1/3 $ Haldane pseudopotential. 
Similar to the fermionic case (see also~\cite{wang:2015}), we conjecture that for $ \kappa > 1/2 $ and $ \lambda $ small enough, the first exited (bulk) state of the Hamiltonian $ H_\Lambda $ is found as a ground state $ E_0( \mathcal{D}_{l,r}^{(2)}) $  in the invariant subspace $ \mathcal{D}_{l,r}^{(2)}  $. This state resembles a bound state consisting of a particle-hole pair separated by $ 2(m-1) $ sites. Within this sector, we compute in the physical parameter regime of small $ |\lambda| $ an explicit expression for the minimum up to order $ |\lambda|^2 $. Moreover, we provide a variational state from $\cD_{l,r}^{(2)}$ that agrees with this expansion to $\mathcal{O}(|\lambda|^2)$. This state is defined in terms of the squeeze Tao-Thouless state $\vp_r$ and the excited state $\eta_r$ defined as
	\begin{equation}
		\eta_r  := - \tfrac{\overline{\lambda}}{\sqrt{2}} \beta_{r-1}  | 10\rangle \otimes \varphi_{r-1} + | 0200\rangle \otimes \varphi_{r-2} = \eta_2 \otimes \varphi_{r-2} - \frac{\beta_{r-1} |\lambda|^2}{2} \ | 10\rangle \otimes \eta_{r-1} .
	\end{equation}
Note that $\eta_r$ is the mirror of the excited state $\eta_r^{(2)}$ considered for the martingale method in Section~\ref{sec:MM}. Similar to that case, $\eta_r$ is orthogonal to $\varphi_r $ and its norm satisfies $ \|  \eta_r  \| = \sqrt{\beta_{r-1}} \|   \varphi_r  \| $.

\begin{theorem}[Exited States]
	For any $ l, \, r \geq 3 $ and $ \kappa > 1/2 $, and all sufficiently small $ |\lambda| $: 
	\begin{equation}\label{eq:exstateasym}
		\min_{m\in \mathbb{N} }\ E_0(\mathcal{D}_{l,r}^{(m)} ) =1 - \frac{2\kappa}{2\kappa-1} \left| \lambda \right|^2 + \mathcal{O}( \left| \lambda \right|^4).
	\end{equation}
	Moreover, a variational state whose energy agrees  to $  \mathcal{O}( \left| \lambda \right|^2) $ with the right side is $  \varphi_l \otimes | 0 \rangle \otimes \psi\in \mathcal{D}_{l,r}^{(2)}$ where 
	$$
	\psi =\left(  | 101 \rangle - \frac{\sqrt{2} \kappa \lambda}{2\kappa -1 }   | 020 \rangle \right)  \otimes  \varphi_r - \frac{\sqrt{2} \lambda}{2\kappa -1 }  \, | 101 \rangle  \otimes \eta_r  .
	$$
\end{theorem}
\begin{proof}
	We first discuss the lower bound on the ground state energy of $H_\Lambda $ restricted to the invariant closed subspace $ \mathcal{D}_{l,r}^{(m)}  $. This is estimated by dropping all terms in the Hamiltonian aside from those acting on the interval $[2m-3,2m+3]$ for $ m \geq 2 $, i.e.
	\begin{equation}\label{eq:lowbdexited}
		H_\Lambda \geq \sum_{x=2m-3}^{2(m+1)} n_x n_{x+1} + \kappa  \sum_{x=2(m-1)}^{2(m+1)} q_x^* q_{x}  ,
	\end{equation}
	and dropping all terms except those acting on the interval $[1,5]$ when $ m = 1 $. 
	
	In the case $m=2$, the restriction of $H_{[1,7]}$ to $D_{l,r}^{(2)}$ is unitarily equivalent to a block matrix of size $6+3=9$. More precisely, the $3\times 3$ block corresponds to the restriction onto the cyclic subspace generated from $(10)(1)(10)(02)$, i.e. $ \spa\{ |1011002\rangle ,  |0201002\rangle ,  |0120002\rangle\} $, given by
	$$
	H_{3\times 3} = \left[\begin{matrix} 1+\kappa |\lambda|^2 & - \sqrt{2} \kappa \overline{\lambda} & 0 \\ 
		- \sqrt{2} \kappa \lambda &  2 \kappa (1 + |\lambda|^2) & - 2 \kappa \overline{\lambda} \\
		0 & - 2 \kappa\lambda & 2 (1+\kappa)  \end{matrix}\right] .
	$$
	If $ \kappa > 1/2 $  it follows by second order perturbation theory that to order $ |\lambda|^2 $ the lowest eigenvalue in such a block is
	$$
	\inf \spec(H_{3\times 3}) = 1 - \frac{\kappa}{2\kappa-1} |\lambda|^2 + \mathcal{O}\left( \left| \lambda \right|^4\right) .
	$$ 
	The corresponding $6\times 6$-block results from restricting $H_{[1,7]}$ to the subspace generated by $(10)(1)(10)_2$, i.e. $\spa\{ |1011010\rangle ,  |0201010\rangle ,  |0120010\rangle , |1010200\rangle ,  |0200200\rangle,  |1002100\rangle \}$, and produces:
	$$
	H_{6\times 6} =\left[\begin{matrix} 1+\kappa |\lambda|^2 & - \sqrt{2} \kappa \overline{\lambda} & 0 & - \sqrt{2} \kappa \overline{\lambda}  & 0 & 0  \\ 
		- \sqrt{2} \kappa \lambda &   \kappa (2 + 3 |\lambda|^2) & - 2 \kappa \overline{\lambda}  & 0 & - \sqrt{2} \kappa \overline{\lambda}  & 0 \\
		0 & - 2 \kappa\lambda & 2 (1+\kappa) & 0 & 0 & 0 \\
		- \sqrt{2} \kappa \lambda & 0 & 0  &  \kappa (2 + 3 |\lambda|^2)  &  - \sqrt{2} \kappa \overline{\lambda}  &  - 2 \kappa \overline{\lambda}   \\
		0 & - \sqrt{2} \kappa \lambda & 0 & -  \sqrt{2} \kappa \lambda & 4\kappa & 0 \\
		0 & 0 & 0 & - 2 \kappa \lambda & 0 & 2(\kappa +1) \end{matrix}\right] .
	$$
	Second order perturbation theory then yields 
	\[ \inf \spec( H_{6\times 6}) = 1 - \frac{2\kappa}{2\kappa-1} \left| \lambda \right|^2 + \mathcal{O}( \left| \lambda \right|^4) \] 
	for $ \kappa > 1/2 $, which is clearly smaller than $ \inf \spec(H_{3\times 3}) $.
	
	The analysis for $m>2$ and $m=1$ follows similarly. For $m>2$, the restriction of $H_{[2m-3,2m+3]}$ to $D_{l,r}^{(m)}$ produces a $6+3+3+1=13$ dimensional block matrix. In addition to the two blocks discussed above, there is a $ 1 \times 1 $ block corresponding to restricting to the vector $ |00 11002 \rangle $ with value $ 1 $, and an additional copy $H_{3\times 3}$ corresponding to the invariant subspace $ \spa\{ |0011010\rangle ,  |0010200\rangle ,  |0002100\rangle\} $.
	
	For $ m = 1 $ the restriction of $H_{[1,5]}$ to  $ \mathcal{D}_{l,r}^{(1)}  $ is unitarily equivalent to $ H_{3\times 3} $. This finishes the proof of the fact that the right side in \eqref{eq:exstateasym} is a lower bound on the ground state energy of $ H_\Lambda $ restricted to $ \mathcal{D}_{l,r}^{(m)}  $  for all $ m \in \mathbb{N} $.
	
	For an upper bound we use the Rayleigh-Ritz principle to write
	$$E_0(\mathcal{D}_{l,r}^{(2)})   \leq \langle \varphi_l \otimes \langle 0 | \otimes \psi , H_\Lambda \varphi_l \otimes | 0 \rangle \otimes \psi \rangle / (\| \varphi_l \|^2  \| \psi\|^2) = \langle \psi , H_{\Lambda_r^{(2)}}  \psi \rangle / \| \psi \|^2 $$ 
	and explicitly compute the expectation value.
\end{proof} 

Similar to \cite{NWY:2020} (see also~\cite{MBR:2020}), one can find plenty of many-body scars higher up in the spectrum of $ H_\Lambda $ or $ H_\Lambda^\textrm{per} $. To produce such an eigenstate with low complexity, one places a sufficient number of voids around a local excitation to shield it on either side from a Tao-Thouless ground-state~\eqref{TT}. The particular form of the hopping term then prevents these excitations from move across the void-boundary. For example, 
\be \label{scarstate}
\psi = \vp_l \otimes \ket{01100}   \otimes \vp_r^{(1)}, 
\ee
is an eigenstate of $ H_\Lambda $ with energy $ 1 $, filling fraction $ \nu=1/2$, and Schmidt rank $ 2$.  Inserting the string $ (0110)_n(0) $ produces an eigenstate with energy $ n $. Up to potentially increasing the length of the void-string, in the same manner one can produce other eigenfunctions with local excitations arising from locally confined configurations generated by a finite number of substitution rules. For example, the cyclic subspace of the configuration $ (003000) $ is two-dimensional and consists of the additional configuration $ (011100) $. The eigenvectors corresponding to this invariant subspace constitute local excitations with strictly positive energies which can be similarly inserted in \eqref{scarstate} to create excited states of the larger system.

\minisec{Acknowledgements }

{\small This work was supported by the DFG under EXC-2111--390814868.}

\bibliographystyle{abbrv}  
\bibliography{bosonic} 

\bigskip
\bigskip
\bigskip
\bigskip

\noindent Simone Warzel\\
Munich Center for Quantum Science and Technology, and\\
Zentrum Mathematik  \& Department of Physics, TU M\"{u}nchen\\
85747 Garching, Germany\\
\verb+warzel@ma.tum.de+\\

\noindent Amanda Young\\
Munich Center for Quantum Science and Technology, and\\
Zentrum Mathematik, TU M\"{u}nchen\\
85747 Garching, Germany\\
\verb+young@ma.tum.de+\\

\end{document}